\documentclass[12pt,lettersize,journal,onecolumn]{IEEEtran}
\usepackage{amsmath,amsfonts}
\usepackage{amsmath,amsfonts}
\usepackage{algorithmic}
\usepackage{algorithm}
\usepackage{array}
\usepackage{textcomp}
\usepackage{stfloats}
\usepackage{url}
\usepackage{verbatim}
\usepackage{graphicx}
\usepackage{subfigure}
\usepackage{cite}
\usepackage{mathrsfs}
\usepackage{xcolor}
\usepackage{amssymb}
\usepackage{amsthm}
\usepackage{multirow}
\usepackage{threeparttable}
\usepackage{booktabs}
\usepackage{epstopdf}
\usepackage{stmaryrd}
\usepackage{bm}

\newtheorem{theorem}{Theorem}
\newtheoremstyle{noparens}%
  {}{}%
  {\itshape}{}%
  {\bfseries}{.}%
  { }%
  {\thmname{#1}\thmnumber{ #2}\mdseries\thmnote{ #3}}
\theoremstyle{noparens}
\newtheorem{theoremNoParens}[theorem]{Theorem}
\newtheorem{lemma}{Lemma}

\newtheorem{definition}{Definition}
\newtheorem*{remark}{Remark}

\usepackage{setspace}
\begin{document}

\title{On the Performance of Low-complexity Decoders of LDPC Codes}
\author{Qingqing Peng, Dawei Yin, Dongxu Chang, Yuan Li,\\ Huazi Zhang, Guiying Yan, Guanghui Wang

\thanks{This work is partially supported by the National Key R\&D Program of China, (2023YFA1009602) (Corresponding author: Dawei Yin)

Qingqing Peng, Dawei Yin, and Dongxu Chang are with the School of Mathematics, Shandong University, Jinan, 250100 China (e-mail:\{pqing, dongxuchang\} @mail.sdu.edu.cn, daweiyin2@outlook.com).

Yuan Li and Huazi Zhang are with Huawei Technologies Co. Ltd., Hangzhou, 310051 China (e-mail: \{liyuan299, zhanghuazi\}@huawei.com).

Guiying Yan is with the Academy of Mathematics and Systems Science, CAS, University of Chinese Academy of Sciences, Beijing, 100190 China (e-mail: yangy@amss.ac.cn).

Guanghui Wang is with the School of Mathematics, Shandong University, Jinan, 250100 China, and also with the State Key Laboratory of Cryptography and Digital Economy Security, Shandong University, Jinan, 250000 China (e-mail: ghwang@sdu.edu.cn).
}
}

\markboth{Journal of \LaTeX\ Class Files,~Vol.~14, No.~8, xxxx~xxxx}%
{Shell \MakeLowercase{\textit{et al.}}: A Sample Article Using IEEEtran.cls for IEEE Journals}

\IEEEpubid{0000--0000/00\$00.00~\copyright~2021 IEEE}


\maketitle
\begin{abstract}
Efficient decoding is crucial to high-throughput and power-sensitive wireless communication scenarios. A theoretical analysis of the performance-complexity tradeoff toward low-complexity decoding is required for a better understanding of the fundamental limits in the above-mentioned scenarios. This study aims to explore the performance of LDPC codes under belief propagation (BP) decoding with complexity constraints. In other words, for a small number of iterations, we present a closed-form lower bound on the bit error rate (BER) of LDPC codes as a function of complexity.
Specifically, for the regular LDPC code ensembles, the dominant term in the order of the lower bound we provide matches that of the upper bound given by Lentmaier.
Furthermore, for irregular LDPC code ensembles, in addition to adopting the approach used for regular codes, we also propose a method to iteratively obtain the lower bound on the average BER.
\end{abstract}


\section{introduction}

Through decades of efforts in pursuit of the Shannon limit, channel coding has seen major breakthroughs in both theory and practice. The low-density parity-check (LDPC) code, employed in 5G New Radio (NR) as a channel coding scheme, has theoretically approached the Shannon limit\cite{mackay1997near}. As a result, communication systems now have high spectral efficiency. However, the development of $6$G introduces new requirements for modern coding schemes, including low latency and high throughput. For instance, in immersive communication scenarios, next-generation channel coding is expected to support peak data rates of 200 Gbps to 1 Tbps, with an energy efficiency target of approximately 1 pJ/bit \cite{itu2023framework,10812743,9200376}. These requirements indicate that LDPC codes are expected to complete decoding with low complexity, highlighting the importance of evaluating their performance under low-iteration decoding constraints.

The performance of LDPC codes under belief propagation (BP) decoding has been a central topic of research. A notable milestone in this area is the introduction of the density evolution technique by Urbanke and Richardson \cite{richardson2002capacity}. They introduced density evolution as a powerful analytical tool for computing the average bit error rate (BER) of LDPC code ensembles. In this asymptotic regime, the average BER can be accurately characterized via iterative message passing on a tree-like computation graph. Furthermore, they demonstrated that the performance of a randomly chosen LDPC code closely approximates the ensemble average with high probability. Subsequent methods, such as Gaussian approximation and extrinsic information transfer (EXIT) chart \cite{910580,957394}, have been developed to simplify the analysis of density evolution. These approaches provide an iterative method for calculating performance, but cannot yield a closed-form expression.

However, directly deriving a closed-form expression for the performance of LDPC codes as a function of the decoding complexity of BP decoding is a challenging task. Some works have attempted to derive closed-form upper and lower bounds on the performance of LDPC codes to facilitate a better understanding of the behavior of LDPC codes. Lentmaier et al. established a closed-form upper bound on the BER for regular LDPC codes under $O(\log N)$ iterations, where $N$ is the block length \cite{lentmaier2005analysis}. Sahai et al. explored the relationship between decoding performance and the size of the largest neighborhood \cite{sahai2010general}.

In this study, we aim to establish a direct relationship between the decoding complexity and a lower bound of the average BER $\overline{P}_e$ of the LDPC code ensemble. Unlike Sahai’s work, we use a more direct metric—the number of iterations $l$—to measure the decoding complexity, rather than the size of the largest neighborhood, which only indirectly reflects complexity and requires additional estimation and scaling across different code ensembles. Moreover, in order to meet high throughput requirements, we focus on the performance at low iteration numbers. Specifically, our objective is to identify a function $f(l)$ that serves as a lower bound on the average BER, i.e., $\overline{P}_e \ge f(l)$. Unlike the iterative equations derived from density evolution, the function \( f(l) \) directly provides an explicit characterization of the performance lower bound of LDPC codes.

A better understanding of the lower bound on the performance of LDPC codes can offer valuable insights in several ways. First, it reveals the tradeoff between performance and decoding complexity. For instance, given a target error probability, one can estimate the minimum required decoding resources, such as decoding complexity. Conversely, when decoding complexity is constrained, lower bounds allow us to determine the best achievable performance under such limitations, which is critical for assessing the practical viability of a given code. Moreover, lower bounds can provide theoretical guidance for the structural design and optimization of LDPC codes.
\subsection{Our Contributions}
In this study, we derive a unified closed-form lower bound on the average BER of LDPC code ensembles, which is universally valid for both regular and irregular ensembles and for any iteration number $l$ under flooding BP decoding. To this end, we introduce an intermediate quantity: the expected minimum Hamming weight of the corresponding root-constrained tree code ensemble, which depends on the structure of the computation graph corresponding to the decoding process. The proposed approach consists of two main steps. First, we establish a direct relationship between the lower bound on the ensemble-average BER and the expected minimum Hamming weight of the root-constrained tree code ensemble. Then, we compute the functional relationship between the expected minimum Hamming weight and the number of decoding iterations $l$ separately for both regular and irregular LDPC codes. This framework enables a rigorous characterization of how the average BER evolves with the number of iterations, thereby fulfilling the main objective of this work.

Furthermore, for regular LDPC code ensembles with variable node degree $J$ and check node degree $K$, we show that when $l \le \log_{(J-1)(K-1)}\left(1-\frac{K}{J(K-1)}\right)N$, the BER lower bound can be expressed as $h2^{-a 2^{b l}}$ for the binary symmetric channel (BSC), the binary erasure channel (BEC), and the binary-input additive white Gaussian noise (BI-AWGN) channel, where $h, a > 0$ and $b \ge 0$ are constants determined by the specific channel and coding scheme. The term $b$ serves as the dominant factor determining the order of the average BER lower bound. We further show that under the condition $l \le \theta_1 \log_{(J-1)^5 (K-1)^3} N^2$ with $\theta_1 < 1$, the dominant term \( b \) matches the upper bound on the BER of LDPC codes provided by Lentmaier et al. \cite{lentmaier2005analysis}. This indicates that the derived closed-form lower bound is tight under these conditions and is also consistent with the double-exponential decay of BER with respect to the number of iterations, as conjectured by Grover et al. \cite{grover2011towards}.

For irregular LDPC code ensembles, we utilize the maximum degrees of variable and check nodes to adapt the lower bound results obtained for regular codes, thereby deriving a corresponding lower bound for the irregular case. To achieve a tighter bound, we propose an alternative iterative computational method that differs from traditional density evolution techniques. This bound is established by computing the expected minimum Hamming weight of the corresponding root-constrained tree code ensemble. Since this analysis is more complex for irregular codes, we reformulate the problem from a graph-theoretic perspective: namely, determining the number of distinct neighbors at a given distance from a specific vertex.

Taken together, the main contributions of this work are: (\romannumeral1) proposing a unified analytical framework applicable to both regular and irregular LDPC code ensembles, establishing a direct relationship between the number of decoding iterations and a lower bound on the average BER of LDPC code ensemble, based on the structure of the computation graph; (\romannumeral2) demonstrating the tightness of the proposed closed-form lower bound for regular LDPC code ensembles at low iteration numbers, whose main term matches that of the upper bound in  \cite{lentmaier2005analysis}; and
(\romannumeral3) beyond the unified analytical framework, for irregular LDPC code ensembles, providing an iterative computational approach based on a graph-theoretic reformulation to obtain a tighter lower bound, which serves as a complementary result to the theoretical analysis.

\subsection{Paper Organization}
The remainder of the paper is organized as follows. Section \ref{sec1} introduces LDPC codes and BP decoding. Section \ref{sec_both} derives a lower bound on the average BER of regular LDPC ensembles based on the expected minimum Hamming weight of the corresponding root-constrained tree code ensemble. Section \ref{sec2} establishes a direct relationship between the number of iterations
$l$ and the lower bound on the average BER for regular LDPC ensembles. Section \ref{sec3} extends the analytical framework to irregular LDPC ensembles and introduces an alternative method for computing the lower bound on the average BER. Sections \ref{simulation} and \ref{sec4} present the simulation results and conclude the paper with a summary of the main findings, respectively.

\section{PRELIMINARIES}
\label{sec1}
This Section will review the definitions of LDPC codes and BP decoding.
\subsection{Graphs and LDPC Codes}

An undirected graph $G = (U,E)$ is defined as a set of vertices $U$ and a set of edges $E\subseteq U \times U$. The degree of a vertex $u \in U$ refers to the number of edges connected to $u$. A walk of length $k$ in $G$ is a sequence of vertices $u_1, u_2, \ldots, u_{k+1}$ in $U$ such that $(u_i,u_{i+1}) \in E$ for all $i=1,\ldots,k$. If all vertices in this sequence are distinct, the walk is called a path. The distance between two vertices $u$ and $u'$ in $G$, denoted by $dist({u, u'})$, is defined as the length of the shortest path connecting them. If there does not exist a path connecting $u$ and $u'$, then $dist({u, u'}) = \infty$.

A graph $G = (V \cup C, E)$ is called bipartite if $V \cap C = \emptyset$ and every edge in $E$ connects a vertex in $V$ to a vertex in $C$.  
The Tanner graph of an LDPC code is a bipartite graph in this sense, where $V=\{v_1, \ldots, v_N\}$ and $C= \{c_1, \ldots, c_M\}$ are referred to as variable nodes and check nodes, respectively, and $N$ and $M$ represents the numbers of variable and check nodes. The Tanner graph provides a direct graphical representation of the sparse parity-check matrix $H \in \mathbb{F}_2^{M \times N}$ associated with the LDPC code. Specifically, an edge $(v_i, c_j) \in E$ exists if and only if $H(j,i) = 1$. An illustrative example of the parity-check matrix and its corresponding Tanner graph for an LDPC code of length 4 is shown in Fig.~\ref{fig1}.
 \begin{figure}[htbp]
\centering
\subfigure[parity-check matrix $H$]{
\raisebox{0.65\height}{
\includegraphics[width=1.3in]{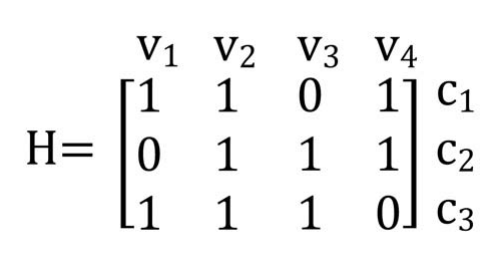}
}
\hfill
}
\subfigure[Tanner graph $G_1$]{
\includegraphics[width=1.3in]{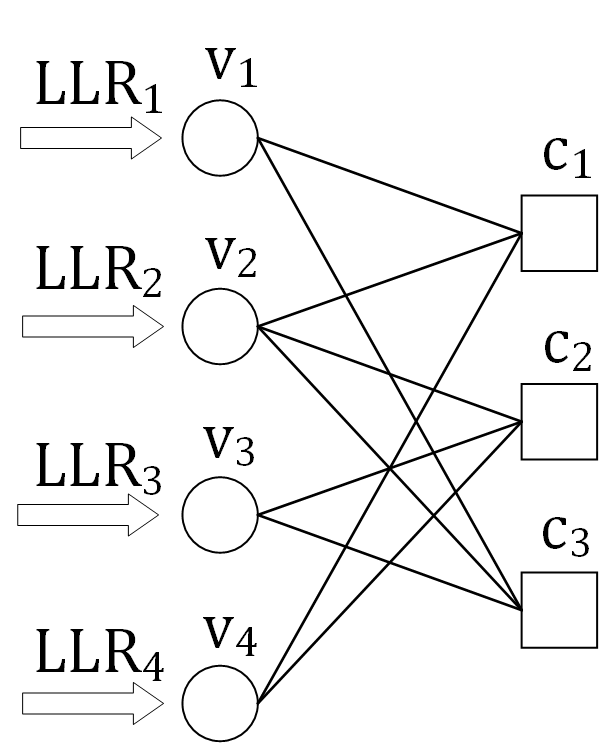}
}
\centering
\caption{Parity-check matrix and Tanner graph for an LDPC code with a code length of $4$. In the Tanner graph, circles represent variable nodes, and squares represent check nodes.  The $i$-th variable node receives channel message LLR$_i$}.
\label{fig1}
\end{figure}

\subsection{The Ensemble of LDPC Codes}
\label{submodel}
The normalized degree distributions from the node perspective of an LDPC code are usually described by two polynomials:
\begin{align}
L(x) = \sum_{d=1}^{\infty} L_d x^d,\;\;\;\;\;\;\;\;\;\;\;\;\; R(x) = \sum_{d=1}^{\infty} R_d x^d,
\end{align}
where $L_d$ denotes the proportion of variable nodes with a degree of $d$, and $R_d$ denotes the proportion of check nodes with a degree of $d$. Correspondingly, the degree distributions from the edge perspective are given by:
\begin{equation}
    \lambda(x) = \sum_{d=1}^{\infty} \lambda_d x^{d-1} =\frac{\sum_{d=1}^\infty dL_dx^{d-1}}{\sum_{d=1}^\infty dL_d},
\end{equation}
\begin{equation}
    \rho(x) = \sum_{d=1}^{\infty} \rho_d x^{d-1}=\frac{\sum_{d=1}^\infty dR_dx^{d-1}}{\sum_{d=1}^\infty dR_d} ,
\end{equation}
where $\lambda_d$ is the fraction of edges that connect to variable nodes of degree $d$, and $\rho_d$ is the fraction of edges that connect to check nodes of degree $d$.


We next describe the construction of \text{LDPC}$(N, L, R)$ \cite{dehghan2018tanner}. In this ensemble, the degrees of variable nodes are sampled according to the degree distribution $L(x)$, while the degrees of check nodes are sampled according to $R(x)$. For each node, a bin is created containing a number of cells equal to the node's degree. Random perfect matchings are then created between the cells on the variable node side and those on the check node side. Each matching yields a so-called configuration in which the matched cells on both sides of the graph are connected by an edge. Throughout this paper, we assume that configurations are selected uniformly at random. Corresponding to each matching (configuration), we construct a bipartite graph such that if there is an edge between two cells, then we place an
edge between the corresponding nodes (bins) in the bipartite graph. The set of all such bipartite graphs forms the ensemble $\mathcal{G}$. Note that $\mathcal{G}$ includes bipartite graphs with parallel edges. The \text{LDPC}$(N, L, R)$ ensemble is then defined by removing all bipartite graphs in $\mathcal{G}$ that contain parallel edges. Given that the bipartite graphs constructed from random configurations contain no parallel edges, the distribution of bipartite graphs in \text{LDPC}$(N, L, R)$ is uniform.

\subsection{BP Decoding and Computation Graph}
\begin{figure*}[!t]
\centering
\includegraphics[width=6in]{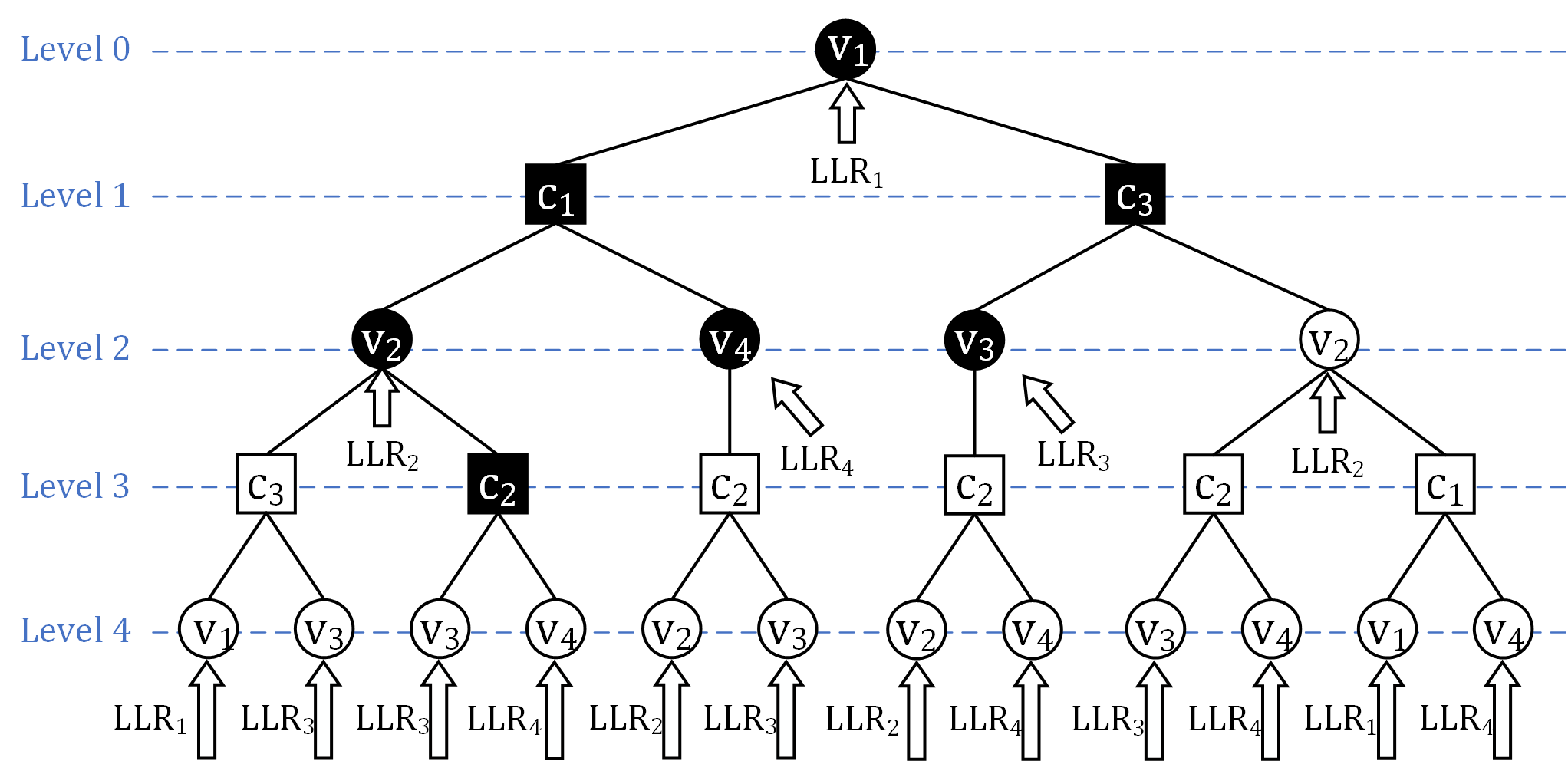}
\caption{A computation graph of height $4$ for $v_1$ in Fig. \ref{fig1}. Black nodes represent the first occurrence of a node, while white nodes indicate subsequent occurrences.}
\label{fig2}
\end{figure*}

One effective decoder for LDPC codes is BP decoding. In a Tanner graph, each variable node receives a message from the channel, and BP decoding is achieved by passing messages along the edges. Assume that $\boldsymbol{s}= \{s_1,\ldots, s_N\} \in \{0,1\}^N$ are transmitted through binary-input memoryless channels, and  $\boldsymbol{y} = \{y_1,\ldots, y_N\}$ are received signals. Let
\begin{equation}
\text{LLR}_i = \ln \frac{p(y_i|s_i=0)}{p(y_i|s_i=1)}
\end{equation}
denote the initial log-likelihood ratio (LLR) for the $i$-th variable node, and LLR$_i$ can also be regarded as the channel message of the $i$-th variable node, as illustrated in Fig. \ref{fig1}.

The message-passing process of BP decoding consists of variable-to-check (V2C) message updates and check-to-variable (C2V) message updates. The message update rule of V2C is
\begin{equation}
\label{v2c}
B_{v_i \to c_j}^{(l)} = \text{LLR}_i + \sum_{\alpha \in \mathcal{N}(v_i) \backslash c_j}B_{\alpha \to v_i}^{(l-1)},
\end{equation}
and the message update rule of C2V is
\begin{equation}
\label{c2v}
B_{c_j \to v_i}^{(l)} = 2 \tanh^{-1}\left(\prod_{\beta \in \mathcal{N}(c_j) \backslash v_i} \tanh \left(\frac{B_{\beta \to c_j}^{(l)}}{2}\right)\right),
\end{equation}
where $l$ is the number of iterations, $\mathcal{N}(v_i)$ denotes the nodes connected directly to node $v_i$, $v_i \to c_j$ means from variable node $v_i$ to check node $c_j$ and $c_j \to v_i$ means from check node $c_j$ to variable node $v_i$. The initial message $B_{v_i \to c_j}^{(0)}$ and $B_{c_j \to v_i}^{(0)}$ is $0$. After $l$ iterations, the output message of variable node $v_i$ is
\begin{equation}
\label{v_all}
B_{v_i}^{(l)} = \text{LLR}_{i} +\sum_{\alpha \in \mathcal{N}(v_i)}B_{\alpha\to v_i}^{(l)}
\end{equation}

Now, let us focus on the decoding process of a single variable node, which can be described using a computation graph \cite[Chapter 3.7.1]{richardson2008modern}. For example, consider the decoding of the node \(v_1\) in Fig. \ref{fig1}. The output message of \(v_1\) is obtained by combining the messages received from its neighboring check nodes \(c_1\) and \(c_3\), along with the channel message LLR$_1$, according to (\ref{v_all}). Furthermore, the outgoing massage of $c_1$ is a function of the messages arriving from \(v_2\) and \(v_4\), following (\ref{c2v}). Similarly, the dependencies among these messages can be represented using a graph, resulting in the tree-like computation graph rooted at \( v_1 \) after $2$ iterations, as illustrated in Fig. \ref{fig2}. In this example, $v_1$ receives four channel messages-$\text{LLR}_1, \text{LLR}_2, \text{LLR}_3,$ and $\text{LLR}_4$-which correspond to the number of distinct variable nodes in the computation graph. This implicitly suggests that the number of distinct variable nodes in the computation graph may affect decoding performance. A more detailed analysis will be presented later.

Let $v$ be a variable node in the Tanner graph $G$, and let $k$ be a non-negative even integer. We then introduce several definitions related to the computation graph.
\begin{definition}(COMPUTATION GRAPH) The computation graph $T_k(v,G)$ is a tree rooted at $v$ and is constructed as follows:
\begin{itemize}
    \item At level 0, $T_{0}(v,G)$ consists of only the variable node $v$.
    \item For each node $u$ at level $k^\prime$ $(0\le k^\prime< k)$, its children are copies of the neighbors of $u$ in $G$ excluding its parent node in $T_{k’}(v,G)$.
    \item The construction stops when the tree reaches height $k$.
\end{itemize}
\label{ref_computationgraph}
\end{definition}

\begin{definition}(TREE CODE AND MINIMUM HAMMING WEIGHT)
\label{def_1}
The tree code \(\mathcal{C}_k(v,G)\) is the set of all valid codewords on \(T_{k}(v,G)\) \cite[Chapter 3.7.2]{richardson2008modern}. Specifically, \(\mathcal{C}_k(v,G)\) consists of all 0/1 assignments to the variable nodes in \(T_{k}(v,G)\) that satisfy all the constraints in the computation graph. Moreover, the root-constrained tree code \(\mathcal{C}^{0/1}_k(v,G)\) denotes the subset of \(\mathcal{C}_k(v,G)\) where the root node \(v\) takes the value \(0\) or \(1\).

Let \( w_k^1(v,G) \) denote the minimum Hamming weight of the code \(\mathcal{C}_k^{1}(v,G)\).
That is,
\begin{equation}
w_k^1(v,G) = \min_{\boldsymbol{s} \in \mathcal{C}_k^{1}(v,G)} \mathrm{wt}(\boldsymbol{s}),
\end{equation}
where \(\mathrm{wt}(\boldsymbol{s})\) denotes the Hamming weight of a codeword $\boldsymbol{s}$.
\end{definition}

Consider the ensemble LDPC\((N, L, R)\). Let \(\mathfrak{C}_k(N, L, R)\) and \(\mathfrak{C}^{0,1}_k(N, L, R)\) denote the ensembles of tree codes and root-constrained tree codes, respectively, induced by LDPC\((N, L, R)\). They are defined as follows:

\begin{definition}(TREE CODE ENSEMBLE AND EXPECTED MINIMUM HAMMING WEIGHT)
To sample an element from the ensemble, a Tanner graph \(G\) is first drawn uniformly at random from \text{LDPC}\((N, L, R)\). Then, a variable node \(v\) is selected uniformly at random from \(G\), which induces the tree code $\mathcal{C}_k(v,G)$ and root-constrained tree code \(\mathcal{C}^{0/1}_k(v,G)\). The ensemble \(\mathfrak{C}_k(N, L, R)\) (respectively, $\mathfrak{C}_k^{0/1}(N,L,R)$) is defined as the collection of all such \(\mathcal{C}_{k}(v,G)\) (respectively, $\mathcal{C}_k^{0/1}(v,G)$) generated under the above sampling procedure. Furthermore, let \(\overline{w}_k^1(N,L,R)\) denote the expected minimum Hamming weight over the ensemble \(\mathfrak{C}_k^{1}(N,L,R)\), i.e.,
\begin{equation}
    \overline{w}_k^1(N,L,R) = \mathbb{E}_{(v,G):G \in LDPC(N,L,R),v \in G} [w_k^1(v,G)].
\end{equation}
\end{definition}

As an illustrative example, consider the computation graph $T_{4}(v_1,G_1)$ shown in Fig.~\ref{fig2}. In Fig.~\ref{fig_add1}, an example codeword defined on $T_{4}(v_1,G_1)$ is depicted. Solid and hollow circles are used to represent variable nodes assigned the values 1 and 0, respectively.

\begin{figure}[!t]
\centering
\includegraphics[width=3in]{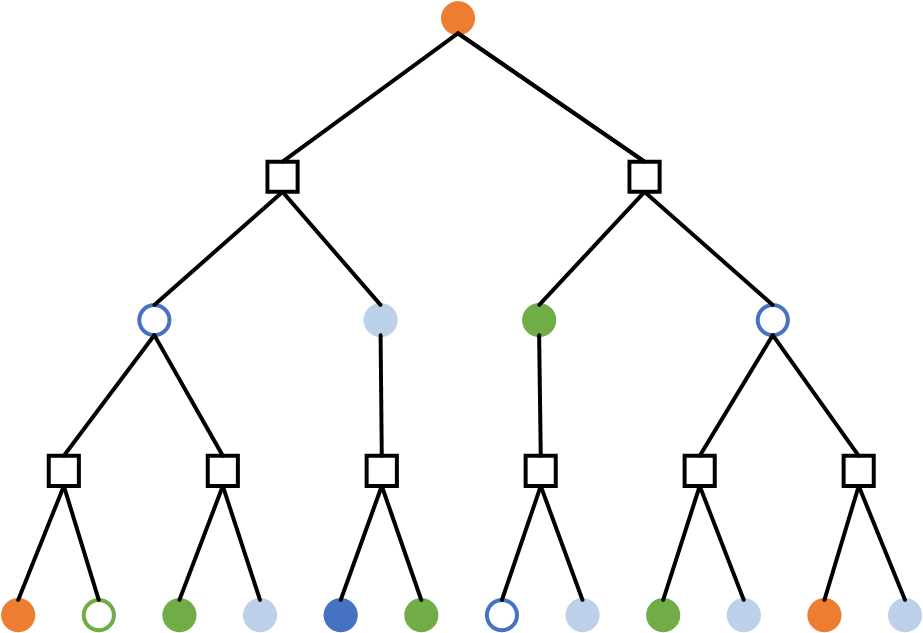} 
\caption{An example codeword from $\mathcal{C}_4(v_1,G_1)$, where solid and hollow circles represent variable nodes assigned the values 1 and 0, respectively. Note that $T_{4}(v_1,G_1)$ contains repeated variable nodes; circles with the same color indicate the same variable node, which takes the same value (either 0 or 1) across all occurrences.}
\label{fig_add1}
\end{figure}

\subsection{Big-O Notation}
We use standard Bachmann-Landau notation in this paper. For any non-negative real-valued functions $f(x)$ and $g(x)$, the notation $f(x) = O(g(x))$ (or equivalently $f(x) \le O(g(x))$) implies that for sufficiently large $x$, $f(x)$ is bounded above by $\theta g(x)$, where $\theta$ is a positive constant. 

\section{performance bounds of ldpc code ensembles given minimum Hamming weight}
\label{sec_both}
In this section, we derive a lower bound on the average BER of LDPC code ensembles as a function of the expected minimum Hamming weight of the corresponding root-constrained tree code ensemble, which is governed by the decoding complexity, i.e., the number of iterations. The main theoretical result addressing this objective is presented in Theorem \ref{theo_lowerbound}.

Theorem~\ref{theo_lowerbound} plays a crucial role in establishing the relationship between a lower bound of the average BER and the number of decoding iterations across different LDPC code ensembles. In particular, Sections~\ref{sec2} and~\ref{sec3} derive the relationship between the number of iterations and the expected minimum Hamming weight for regular and irregular ensembles, respectively. Together, these results provide a unified analytical framework that characterizes the dependence of the lower bound of the average BER on the number of iterations through the intermediate quantity---the expected minimum Hamming weight.


\begin{theorem}
\label{theo_lowerbound} 
Consider the ensemble \text{LDPC}$(N, L, R)$ and the decoder performs $l$ iterations of BP decoding. If the expected minimum Hamming weight $\overline{w}_{2l}^1(N,L,R)$ over the ensemble \(\mathfrak{C}_{2l}^{1}(N,L,R)\) is $w$, then the average BER $\overline{P}_e$ of the LDPC code ensemble satisfies
\begin{equation}
\label{lm11}
    \overline{P}_e^{BI-AWGN} \ge Q\left(\sqrt{\frac{w}{\sigma^2}}\right)
\end{equation}
for the BI-AWGN channel, 
\begin{equation}
    \overline{P}_e^{BSC} \ge q^{\frac{w+1}{2}}
\end{equation}
for the BSC, 
\begin{equation}
    \overline{P}_e^{BEC} \ge \epsilon^{w}
\end{equation}
for the BEC, where $Q(x)=\int_x^\infty\frac{1}{\sqrt{2\pi}}e^{-\frac{x^2}{2}}\mathrm{d}x$ is the complementary standard normal Cumulative Distribution Function, $\sigma^2$ denotes the variance of the noise in the Gaussian channel, $q$ denotes the crossover probability in the BSC, and $\epsilon$ denotes the erasure probability in the BEC.
\end{theorem}

\begin{proof}
It suffices to prove the result for the BI-AWGN channel. We first establish a stronger version of the statement: for any Tanner graph \( G \in \)LDPC\((N, L, R)\) and any variable node \( v \) in \( G \), the error probability of $v$ after $l$ iterations of BP decoding is lower bounded as 
\begin{equation}
 P_v^{BI-AWGN} \ge Q\left(\sqrt{\frac{w_{2l}^1(v,G)}{\sigma^2}}\right),   
\end{equation}
where $w_{2l}^1(v,G)$ is the minimum Hamming weight of the code $\mathcal{C}_{2l}^1(v,G)$ as defined in Definition \ref{def_1}.

Consider the code \(\mathcal{C}_{2l}(v,G)\) for transmission. Let \(\boldsymbol{s}\) be a codeword of weight \(w_{2l}^1{(v,G)}\) such that $v$ takes the value 1, and assume that the transmitted codeword is either the all-zero codeword \(\boldsymbol{0}\) or \(\boldsymbol{s}\), each occurring with probability \(1/2\). Transmission takes place over a BI-AWGN channel, and the received vector is denoted by \(\boldsymbol{y}\). Let \(P_v^{\mathcal{C}_{2l}(v,G)}\) denote the probability that variable node \(v\) is in error after \(l\) iterations of BP decoding.

We claim that $p_v^{BI-AWGN} = P_v^{\mathcal{C}_{2l}(v,G)}$. This equality follows from the observation in \cite[Chapter 3]{richardson2008modern} that the error probability of $v$ depends on the structure of the computation graph rooted at $v$ but is independent of the specific transmitted codeword. Moreover, by the Definition \ref{ref_computationgraph}, the code ${\mathcal{C}_{2l}(v,G)}$ shares the same computation graph of height $2l$, rooted at $v$, as the code $G$. This suggests that the claim holds.

We now focus on the error probability $P_v^{\mathcal{C}_{2l}(v,G)}$. Let $Y_0$ and $Y_1$ denote the sets of received sequences $\boldsymbol{y}$ such that, when any element from $Y_0$ (respectively, $Y_1$) is input into the BP decoder, the decoder output belongs to $\mathcal{C}^0_{2l}(v,G)$ (respectively, $\mathcal{C}^1_{2l}(v,G)$). Denote by $Y_0'$ (respectively, $Y_1'$) the complement of $Y_0$ (respectively, $Y_1$), so that the received space $\mathbb{R}^{n}=Y_0 \cup Y_0'=Y_1\cup Y_1'$. Then the error probability of node \( v \) is 
\begin{equation}
\begin{split}
\label{p_v}
&P_v^{BI-AWGN}=P_v^{\mathcal{C}_{2l}(v,G)} \\
&= p(\boldsymbol{0}) \int_{\boldsymbol{y} \in Y_0'} p(\boldsymbol{y}|\boldsymbol{0})\mathrm{d}\boldsymbol{y}+ p(\boldsymbol{s}) \int_{\boldsymbol{y} \in Y_1'} p(\boldsymbol{y}|\boldsymbol{s})\mathrm{d}\boldsymbol{y} \\
&=\frac{1}{2}\left(\int_{\boldsymbol{y} \in Y_0'} p(\boldsymbol{y}|\boldsymbol{0})\mathrm{d}\boldsymbol{y} + \int_{\boldsymbol{y} \in Y_1'} p(\boldsymbol{y}|\boldsymbol{s})\mathrm{d}\boldsymbol{y}\right) \\
&\overset{(e1)}{=}\frac{1}{2}\Bigg(\int_{\boldsymbol{y} \in Y_1} p(\boldsymbol{y}|\boldsymbol{0})\mathrm{d}\boldsymbol{y} + \int_{\boldsymbol{y} \in Y_0'\cap Y_1'} p(\boldsymbol{y}|\boldsymbol{0})\mathrm{d}\boldsymbol{y} \\
&\quad\quad+\int_{\boldsymbol{y} \in Y_0} p(\boldsymbol{y}|\boldsymbol{s})\mathrm{d}\boldsymbol{y}+\int_{\boldsymbol{y} \in Y_0'\cap Y_1'} p(\boldsymbol{y}|\boldsymbol{s})\mathrm{d}\boldsymbol{y}\Bigg) \\
&\overset{(e2)}{\ge} \frac{1}{2}\left(\int_{\boldsymbol{y} \in Y} \min \{p(\boldsymbol{y}|\boldsymbol{0}),p(\boldsymbol{y}|\boldsymbol{s})\}\mathrm{d}\boldsymbol{y}\right),
\end{split}
\end{equation}
where $p(\boldsymbol{s})$ denotes the probability of codeword $\boldsymbol{s}$ being sent, and $p(\boldsymbol{y}|\boldsymbol{s})$ denotes the transition probability of a binary-input output-symmetric memoryless channel. The validity of $(e1)$ follows directly from the definitions of \( Y_0 \) and $Y_1$, with \( Y_0 \cap Y_1 = \emptyset \) and \( Y_1 \subseteq Y_0' \). The validity of $(e2)$ relies on the partitioning of the received space \( Y \), that is, \( Y = Y_0 \cup Y_1 \cup (Y_0'\cap Y_1') \).

Interestingly, the term 

\begin{equation}
\label{term}
  P_{REP}^{BI-AWGN} \triangleq \frac{1}{2}\left(\int_{\boldsymbol{y} \in Y} \min \{p(\boldsymbol{y}|\boldsymbol{0}),p(\boldsymbol{y}|\boldsymbol{s})\}\mathrm{d}\boldsymbol{y}\right)
\end{equation}
can be considered as the BLER of a code composed of two codewords, $\boldsymbol{0}$ and $\boldsymbol{s}$, with each codeword being sent with a probability of $1/2$, and MAP decoding is used. The distance between $\boldsymbol{0}$ and $\boldsymbol{s}$ is $w_{2l}^1(v,G)$, thus (\ref{term}) corresponds to the BLER of a repetition code with a minimum distance of $w_{2l}^1(v,G)$, whose value is given in \cite{cui2020does} as follows:
\begin{equation}
     P_{REP}^{BI-AWGN} = Q\left(\sqrt{\frac{w}{\sigma^2}}\right)
\end{equation}
This completes the proof of the stronger version of the statement.

Theorem~\ref{theo_lowerbound} can be straightforwardly derived from the following equation: 

\begin{equation}
    \begin{split}
&\mathbb{E}_{G\in \text{LDPC}(N,L,R)}\left[\mathbb{E}_{v\in G}\left[P_v^{BI-AWGN}\right]\right] \\
&\ge \mathbb{E}_{G\in \text{LDPC}(N,L,R)}\left[\mathbb{E}_{v\in G}\left[Q\left(\sqrt{\frac{w_{2l}^1(v,G)}{\sigma^2}}\right)\right]\right] \\
& \overset{(e3)}{\ge} Q\left(\frac{\mathbb{E}_{G\in \text{LDPC}(N,L,R)}\left[\mathbb{E}_{v\in G}\left[\sqrt{w_{2l}^1(v,G)}\right]\right]}{\sigma}\right) \\
& \overset{(e4)}{\ge}  Q\left(\sqrt{\frac{w}{\sigma^2}}\right) 
\end{split}
\end{equation}
The validity of $(e3)$ follows from the fact that the second derivative of $Q(x)$ is positive for $x >0$, together with Jensen’s inequality. The justification for $(e4)$ stems from the fact that \( Q(x) \) is strictly decreasing for \( x > 0 \), and the expected value of $w_{2l}^1(v,G)$ is $w$. 
\end{proof}
\begin{remark}
Here, we also provide an intuitive explanation of Theorem~\ref{theo_lowerbound}’s proof. We construct a new code consisting only of $\boldsymbol{0}$ and $\boldsymbol{s}$. By analyzing the error probability of this newly constructed code, we can derive a lower bound of the error probability of node $v$ in the original code $G$.
\end{remark}
\begin{remark}
Based on the stronger version of Theorem~\ref{theo_lowerbound}, it follows that the proposed result can be applied to both standardized and specific LDPC codes once the corresponding expected minimum Hamming weight (or a computable upper bound thereof) is known. In particular, for well-structured code families such as the LDPC codes specified in the Consultative Committee for Space Data Systems (CCSDS) standard \cite{CCSDS131B3} and the IEEE 802.11n standard \cite{11090080}, the expected minimum Hamming weight can be efficiently estimated either through simulation or directly from the base matrix structure. These estimates can then be substituted into Theorem~\ref{theo_lowerbound} to obtain a practical lower bound on the BER for the given code family. This demonstrates that the analytical framework developed in this work is not only theoretically rigorous but also applicable to real-world standardized LDPC codes.
\end{remark}

The value of \( w_{2l}^1(v,G) \) depends on the structure of \( T_{2l}(v,G) \), which is determined by factors such as the number of decoding iterations \( l \) and the code rate of $\mathcal{C}_{2l}(v,G)$. As can be seen from the proof, Theorem~\ref{theo_lowerbound} holds even when \( l \) is sufficiently large such that \( \mathcal{C}_{2l}(v,G) = G \), or when the code rate exceeds the channel capacity. In these cases, however, the lower bound provided by Theorem~\ref{theo_lowerbound} becomes trivial.

\section{performance bounds of regular ldpc code ensembles}
\label{sec2}

This section first presents Theorem~\ref{theo_regular_lower}, which provides a closed-form lower bound on the average BER of regular LDPC code ensembles as a function of the number of decoding iterations. To provide theoretical justification for this result, we subsequently establish the relationship between the number of decoding iterations and the expected minimum Hamming weight of the corresponding root-constrained tree code ensemble. Together with Theorem~\ref{theo_lowerbound}, this relationship constitutes the analytical foundation for the proof of Theorem~\ref{theo_regular_lower}. The tightness of the derived bound is discussed in Subsection \ref{section_re_B}.

\subsection{Main Result}
\label{regular_ensemble}
In the following, we present the main result for regular LDPC code ensembles.
\begin{theorem}
\label{theo_regular_lower}
 Consider the ensemble \text{LDPC}$(N, x^J, x^K)$ with fixed degrees $J \ge 3$ and $K$, where the decoder performs $l$ iterations of BP decoding. For $N$ is sufficiently large, the average BER $\overline{P}_e$ of the LDPC code ensemble satisfies
\begin{equation}
\begin{split}
  \overline{P}_e^{BI-AWGN} \ge Q\left(\sqrt{\frac{\overline{w}^{UB}_{2l}(N,x^J,x^K)}{\sigma^2}}\right)\ge \frac{e^{1/(\pi+ 2)}}{4} \sqrt{\frac{\pi+2}{\pi}} 2^{-\frac{1}{ \sigma^2\ln2} \overline{w}^{UB}_{2l}(N,x^J,x^K)}
\end{split}  
\end{equation}
for the BI-AWGN channel with noise variable $\sigma^2$,  
\begin{equation}
\begin{split}
    \overline{P}_e^{BSC} \ge q^{\frac{\overline{w}^{UB}_{2l}(N,x^J,x^K)+1}{2}}
\end{split}
\end{equation}
for the BSC with crossover probability $q$, 
\begin{equation}
\label{bec_lower}
\begin{split}
    \overline{P}_e^{BEC} \ge \epsilon^{\overline{w}^{UB}_{2l}(N,x^J,x^K)}
\end{split}
\end{equation}
for the BEC with erasure probability $\epsilon$, where $\overline{w}^{UB}_{2l}(N,x^J,x^K)$ is given by
\begin{equation}
\overline{w}^{UB}_{2l}(N,x^J,x^K) \triangleq\begin{cases}
\frac{J(J-1)^l - 2}{J - 2},&\text{if } l \le \theta_1\log_{(J-1)^5(K-1)^3}N^2\\
N,&\text{if } l> \log_{(J-1)(K-1)}\left(\left(1-\frac{K}{J(K-1)}\right)N\right) \\
\frac{J(K-1)^{l+1}(J-1)^l-K}{(J-1)(K-1)-1},&\text{ otherwise}, \\
\end{cases}
\end{equation}
$\theta_1$ is any constant strictly less than 1.
\end{theorem}

Note that $\overline{w}^{UB}_{2l}(N,x^J,x^K) $ depends on the parameters $l,N,J$, and $K$. Consequently, Theorem~\ref{theo_regular_lower} establishes a lower bound on the ensemble-average BER over different channel models as a function of the number of decoding iterations $l$, which holds for any $l$.

According to Theorem \ref{theo_regular_lower}, when the number of decoding iterations $l$ is constant, the error rate does not vanish as the block length $N$ increases. When the number of iterations $l$ increases as a function of $N$ and \( l \le \theta_1 \log_{(J-1)^5 (K-1)^3} N^2 \), it follows that:
\begin{equation}
    \overline{P}_e^{BI-AWGN} \ge \frac{e^{1/(\pi+ 2)}}{4} \sqrt{\frac{\pi+2}{\pi}} e^{\frac{2}{\sigma^2(J-2)}}2^{-\frac{1}{\sigma^2 \ln 2} \frac{J}{J-2} 2^{l \log_2(J-1)}}.
\end{equation}
\begin{equation}
    \overline{P}_e^{BSC} \ge q^\frac{J-4}{2(J-2)}{}2^{-|\log_2 q|\frac{J}{2(J-2)}2^{l\log_2(J-1)}},
\end{equation}
\begin{equation}
    \overline{P}_e^{BEC} \ge \epsilon^{-\frac{2}{J-2}}2^{-|\log_2 \epsilon|\frac{J}{J-2}2^{l\log_2(J-1)}},
\end{equation}
This implies that the lower bound of the average BER can be expressed as
\begin{equation}
 h2^{-a 2^{bl}}   
\end{equation}
for the BSC, BEC, and BI-AWGN channel,  where $h, a>0$ and $b\ge0$ are constants that depend on the specific channel and coding parameters. The parameter $b$, which is the main term determining the order of the lower bound, is given by $b = \log_2(J-1)$. Similarly, when the number of iterations $l$ satisfies \(\theta_1 \log_{(J-1)^5(K-1)^3} N^2 < l \le \log_{(J-1)(K-1)}\left(\left(1-\frac{K}{J(K-1)}\right)N\right),\) the lower bound can also be expressed as $h2^{-a 2^{bl}}$, with the dominant term is given by
$b = \log_2\left((J-1)(K-1)\right)$.

To prove the result of Theorem~\ref{theo_regular_lower}, we need an intermediate result as discussed below.
\begin{lemma}
\label{lemma_regular_w}
Consider the ensemble LDPC$(N, x^J, x^K)$ with fixed degrees $J$ and $K$, and a positive integer $l$. When $N$ is sufficiently large, then the expected minimum Hamming weight $\overline{w}_{2l}^1(N,x^J,x^K)$ is upper bounded by $\overline{w}^{UB}_{2l}(N,x^J,x^K)$. 
\end{lemma}
\begin{proof}[Proof of Lemma \ref{lemma_regular_w}]
 The detailed proof is provided in APPENDIX \ref{prf_BER}. In the proof, we first select a Tanner graph \(G\) uniformly at random from LDPC\((N, x^J, x^K)\), and then choose a variable node \(v\) uniformly at random from \(G\). We examine the set of codewords in $\mathcal{C}_{2l}^1(v,G)$ with weight $\frac{J(J-1)^l - 2}{J - 2}$ and analyze their characteristics. It is shown that, with high probability, $\mathcal{C}_{2l}^1(v,G)$ contains a codeword with weight less than $\frac{J(J-1)^l - 2}{J - 2}$. For $\mathcal{C}_{2l}^1(v,G)$ that does not contain such a low-weight codeword, we provide an upper bound on $w_{2l}^1(v,G)$ by computing the number of distinct variable nodes in the computation graph $T_{2l}(v,G)$. By combining these two cases, we obtain an upper bound on the expected value of $w_{2l}^1(v,G)$.
\end{proof}
Now, we return to the proof of Theorem~\ref{theo_regular_lower}.

\begin{proof}[Proof of Theorem~\ref{theo_regular_lower}]
We consider the BI-AWGN channel. According to Theorem~\ref{theo_lowerbound}, the average BER \( \overline{P}_e^{BI-AWGN} \) under $l$ iterations BP decoding satisfies
\begin{equation}
  \begin{split}
\overline{P}_e^{BI-AWGN} \ge Q\left(\sqrt{\frac{\overline{w}_{2l}^1(N,x^J,x^K)}{\sigma^2}}\right).
\end{split}  
\end{equation}
According to Lemma~\ref{lemma_regular_w}, we have 
\begin{equation}
   \overline{w}_{2l}^1(N,x^J,x^K) \le \overline{w}^{UB}_{2l}(N,x^J,x^K).
\end{equation}
Since \( Q(x) \) is strictly decreasing for \( x > 0 \), and satisfies the inequality \cite{cote2012chernoff}  
 \begin{equation}
     Q(x) \ge \frac{e^{1/(\pi+ 2)}}{4} \sqrt{\frac{\pi+2}{\pi}}e^{-x^2}, \quad \text{for any } x \ge 0.
\end{equation}
Theorem~\ref{theo_regular_lower} is thus established.
\end{proof}

\subsection{The Discussion on the Tightness of the Lower Bound}
\label{section_re_B}
In this Subsection, we will discuss the asymptotic tightness of the lower bound on the average BER over the regular LDPC code ensemble, as derived in Section \ref{regular_ensemble}. To achieve this goal, we refer to an upper bound on the BER for an LDPC code with tree-like computation graphs, as presented in \cite{lentmaier2005analysis}, and compare it with our lower bound. Although the upper bound in \cite{lentmaier2005analysis} was originally derived under the assumption that \( l = O(\log N) \), it remains valid for any \( l \), provided the computation graph retains a tree-like structure. Therefore, we rewrite it as follows. 

\begin{theoremNoParens}[\cite{lentmaier2005analysis}]
\label{theo_upper}
Consider the iterative decoding of an $(N, J, K)$-regular LDPC code with $J \ge 3$, under the assumption that all computation graphs corresponding to this code with height $2\frac{\log_2N}{\log_2(J-1)(K-1)}$ are tree-like. Assume that $l_0$ is a constant determined by $J$, $K$, and the channel. If the number of iterations $l \in \left(l_0,\frac{\log_2N}{\log_2(J-1)(K-1)}\right)$, then there exists a constant $a_0$ such that the BER $P_e$ is approximately upper-bounded by the inequalities 
\begin{equation}
\label{7_1}
    P_e < 2^{-a_02^{l\log_2(J-1)}}.
\end{equation}
\end{theoremNoParens}

In order to meet the requirements of low latency and high throughput, we consider the asymptotic behavior of the average BER lower bound as $N$ tend to infinity, where the number of iterations $l$ grows with $N$ and satisfies $l\le \log_{(J-1)(K-1)}\left(\frac{JK-J-K}{J(K-1)}N\right)$. This range reflects practical constraints on decoding latency. For example, as reported in \cite{tarver2021gpu}, a stringent latency budget permits only five flooding iterations for a code of length 25,334. In this case, both the lower bound in Theorem \ref{theo_regular_lower} and the upper bound in Theorem \ref{theo_upper} are in the form of $h2^{-a2^{bl}}$, where $h,a >0$ and $b\ge0$ are constants dependent on the channel and coding schemes. Among these parameters, $b$ is the main term determining the order of the bound. We investigate the tightness of the lower bound by dividing the analysis into two regimes according to whether the number of iterations exceeds or falls below $\theta_1 \log_{(J-1)^5(K-1)^3} N^2$. When the number of iterations $l\in (l_0,\theta_1 \log_{(J-1)^5(K-1)^3} N^2)$, the main term \( b \) in both the lower and upper bounds satisfies
\begin{equation}
  b=\log_2(J-1), 
\end{equation}
which indicates that the lower bound is asymptotically tight.

When the number of iterations $l\in (\theta_1 \log_{(J-1)^5(K-1)^3} N^2,\log_{(J-1)(K-1)}(\frac{JK-J-K}{J(K-1)}N))$, the dominant term in the upper bound is $b=\log_2(J-1)$, whereas that in the lower bound is $b=\log_2(J-1)(K-1)$. This indicates that a slight discrepancy persists between the upper bound and lower bounds in their leading asymptotic behavior.

\section{Performance bounds of irregular ldpc code ensembles}
\label{sec3}

For an irregular LDPC code ensemble LDPC$(N, L, R)$ with a bounded maximum degree, such as the quasi-cyclic (QC) LDPC code ensemble, let $J_{max}$ and $K_{max}$ denote the maximum degrees of the variable and check nodes, respectively. Then, the average BER of the regular ensemble $\text{LDPC}(N, x^{J_{max}}, x^{K_{max}})$ provides a lower bound on the average BER of the irregular ensemble LDPC$(N,L,R)$. Specifically, the BER of the irregular ensemble LDPC$(N,L,R)$ can also be expressed in the form $h2^{-a2^{bl}}.$ 

However, when there is a significant disparity between the maximum and minimum degrees, this bound becomes relatively loose. To obtain a tighter lower bound, we also propose an iterative method.

\begin{theorem}
\label{theo_irregular_lower}
Consider the ensemble LDPC$(N, L, R)$. When BP decoding is performed with \( l \) iterations and $N$ is sufficiently large, the average BER \( \overline{P}_e \) over the ensemble satisfies
\begin{equation}
\begin{split}
  \overline{P}_e^{BI-AWGN} \ge Q\left(\sqrt{\frac{\overline{w}^{UB}_{2l}(N,L,R)}{\sigma^2}}\right)\ge \frac{e^{1/(\pi+ 2)}}{4} \sqrt{\frac{\pi+2}{\pi}} 2^{-\frac{1}{ \sigma^2\ln2} \overline{w}^{UB}_{2l}(N,L,R)}
\end{split}  
\end{equation}
for the BI-AWGN channel with noise variable $\sigma^2$,  
\begin{equation}
\begin{split}
    \overline{P}_e^{BSC} \ge q^{\frac{\overline{w}^{UB}_{2l}(N,L,R)+1}{2}}
\end{split}
\end{equation}
for the BSC with crossover probability $q$, 
\begin{equation}
\begin{split}
    \overline{P}_e^{BEC} \ge \epsilon^{\overline{w}^{UB}_{2l}(N,L,R)}
\end{split}
\end{equation}
for the BEC with erasure probability $\epsilon$, where $\overline{w}^{UB}_{2l}(N,L,R)$ is given by
\begin{equation}
\label{irr_eqw}
\overline{w}^{UB}_{2l}(N,L,R)  \triangleq N\bigg(1- \prod_{t=1}^l P_{2t}\bigg),
\end{equation}
\begin{equation}
\label{eq1}
P_{2t}=\sum_{d=1}^\infty L_d \bigg(\widetilde{P}_{2t-1}\bigg)^d,
\end{equation}

\begin{equation}
\label{eq2}
\widetilde{P}_{2t-1}= 
\begin{cases}\widetilde{P}_{2t-2}, &\text{if }l\le l_1\\
\sum_{d=1}^\infty \rho_d \bigg(\widetilde{P}_{2t-2}\bigg)^{d-1},&\text{if }l> l_1
\end{cases}
\end{equation}

\begin{equation}
\label{eq3}
\widetilde{P}_{2t-2} = \sum_{d=1}^\infty \lambda_d \bigg(\widetilde{P}_{2t-3}\bigg)^{d-1}, 
\end{equation}

\begin{equation}
\label{eq4}
\widetilde{P}_{1} =
\begin{cases}
1-\frac{1}{N},&\text{if } l \le l_1\\
\sum_{d=1}^\infty\rho_d\bigg(1-\frac{1}{N}\bigg)^{d-1}, &\text{if }l > l_1
\end{cases}
\end{equation}
$l_1=\theta_1 \log_{(J_{max}-1)^5(K_{max}-1)^3}N^2$ with $\theta_1<1$, $\lambda$ and $\rho$ denote the variable node and check node degree distributions from the edge perspective, respectively, as defined in Section~\ref{submodel}. Moreover, equations (\ref{eq1}), (\ref{eq2}), and (\ref{eq3}) are valid for $t \ge 1$, $t \ge 2$, and $t \ge 2$, respectively.
\end{theorem}

$P_{2t}$ can be computed recursively using equations (\ref{eq1})$\sim$(\ref{eq4}), with $\widetilde{P}_{2t-1},\widetilde{P}_{2t-2}$ acting as intermediate terms in the recursion. The practical interpretations of both $P_{2t}$, $\widetilde{P}_{2t-1}$, and $\widetilde{P}_{2t-2}$ are discussed in detail in Appendix \ref{appendix_irregular_w}, where we show that these quantities correspond to specific conditional probabilities. 

To better justify Theorem \ref{theo_irregular_lower}, we introduce an intermediate Lemma.

\begin{lemma}
\label{lemma_irregular_w}
Consider the ensemble LDPC$(N, L, R)$ and a positive integer $l$. When $N$ is sufficiently large, then the expected minimum Hamming weight $\overline{w}_{2l}^1(N,L,R)$ is upper bounded by $\overline{w}^{UB}_{2l}(N,L,R)$. 
\end{lemma}
\begin{proof}[Proof of Lemma \ref{lemma_irregular_w}]
A detailed proof is provided in APPENDIX \ref{appendix_irregular_w}. Similar to Lemma \ref{lemma_regular_w}, we estimate $\overline{w}_{2l}^1(N,L,R)$ based on the number of distinct variable nodes in the computation graph. However, computing the number of distinct variable nodes in the computation graph is more challenging for irregular codes. To address this, we reformulate this problem as finding the distribution of the shortest path length in a random bipartite graph. Although this problem has been studied in \cite{nitzan2016distance} within the context of random graphs, we further explore it in the context of random bipartite graphs. In our study, we randomly select an LDPC code from the ensemble. Then, we randomly choose a pair of variable nodes from this code. Using graph-theoretic tools, we calculate the distribution of the distance between these two variable nodes. 
\end{proof}

Although Lemma \ref{lemma_irregular_w} is derived under the asymptotic assumption that $N$ is sufficiently large, numerical evidence indicates that the resulting estimates remain accurate even for moderate block lengths. A more detailed justification of this observation is also provided in Appendix \ref{appendix_irregular_w}.

As for the proof of Theorem \ref{theo_irregular_lower}, it is similar to that of Theorem \ref{theo_regular_lower} and can be obtained by combining Theorem \ref{theo_lowerbound} and Lemma \ref{lemma_irregular_w}; thus, it is omitted here for brevity.

\section{simulation}
\label{simulation}
In this section, we performed simulations to examine the relationship between the lower bound of the average BER over the LDPC code ensemble and the number of decoding iterations. Since the lower bound of the average BER exhibits a double-exponential decay with respect to the number of decoding iterations $l$, we define the parameter $\gamma= \log_2(-\log_2(P))$ to capture the magnitude of the error rate $P$, and focus on the relationship between $\gamma$ and the number of iterations.

In Fig. \ref{fig_4}, we illustrate the parameter $\gamma$ corresponding to the lower bound on the average BER of both the regular and irregular LDPC code ensembles under different numbers of decoding iterations, as derived from Theorems \ref{theo_regular_lower} and \ref{theo_irregular_lower}. In contrast, we also plot the $\gamma$ values corresponding to the upper bound on the average BER, as given from Theorem~\ref{theo_upper}, together with those computed from the average BER obtained via density evolution after $l$ iterations. In Figs.~\ref{fig4_1} and~\ref{fig4_2}, we further compare the proposed lower bound with the simulated BER performance of a randomly constructed code generated using the Progressive Edge Growth (PEG) algorithm~\cite{hu2005regular}. The code has a length of 5400, variable node degree 3, and check node degree 4, and is evaluated under both the BEC and the AWGN channel. It can be observed that, under different channel conditions, both the simulated BER and the average BER obtained via density evolution are bounded by the lower and upper bounds. Moreover, at low iterations, both curves are closer to the proposed lower bound, and they gradually approach the upper bound as the number of iterations increases. Figs.~\ref{fig4_3} and~\ref{fig4_4} show the proposed lower bound on the average BER for the irregular LDPC code ensemble \(\text{LDPC}(N, L, R)\), where \(L(x) = 0.4286x^2 + 0.5714x^3\) and \(R(x) = 0.5x^8 + 0.5x^{10}\) correspond to the degree distributions of the 5G NR BG2 with a code rate of $5/7$. The results are compared with the simulated BER of BG2 with a code rate of $5/7$ and a lifting size of 384. At small iteration numbers, both the simulated BER and the average BER obtained via density evolution lie closer to the proposed lower bound across different channel conditions.

\begin{figure}[htbp]
\centering
\subfigure[Bounds on the average BER for the LDPC$(N,x^3,x^4)$ code ensemble over a BEC with $\epsilon = 0.6$.]{\includegraphics[width=0.4\textwidth]{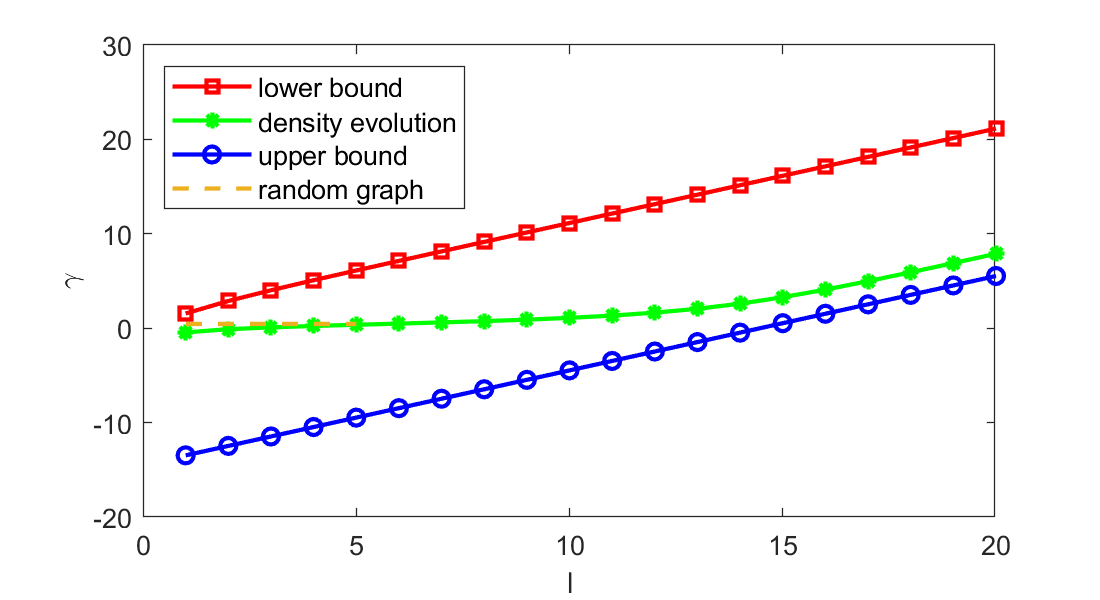}
\label{fig4_1}}
\centering
\subfigure[Bounds on the average BER for the LDPC$(N,x^3,x^4)$ code ensemble over a AWGN channel with $E_b/N_0 = -0.3$dB.]{\includegraphics[width=0.4\textwidth]{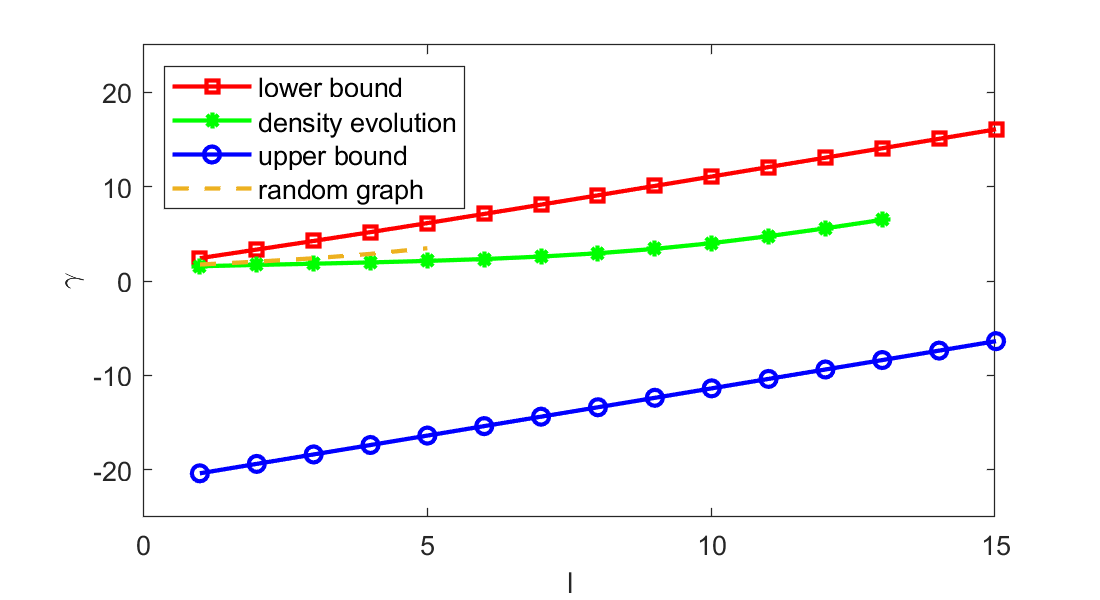}
\label{fig4_2}}
\centering
\subfigure[Bounds on the average BER for the irregular LDPC code ensemble with the degree distribution corresponding to BG2 (rate \(5/7\)) over a BEC with $\epsilon = 0.1$.]{\includegraphics[width=0.4\textwidth]{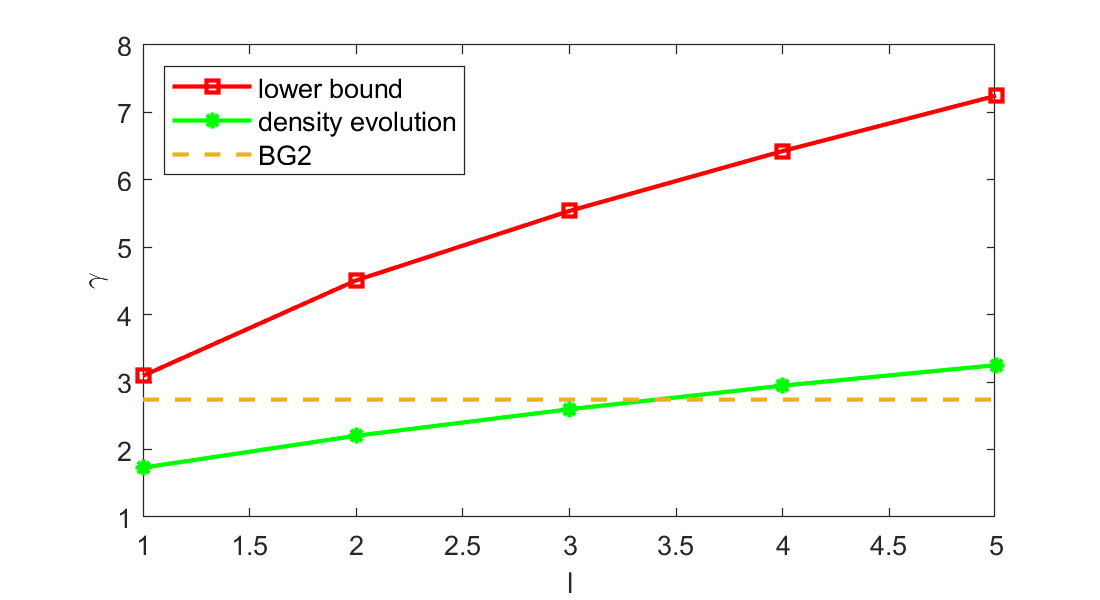}
\label{fig4_3}
}
\centering
\subfigure[Bounds on the average BER for the irregular LDPC code ensemble with the degree distribution corresponding to BG2 (rate \(5/7\)) over a AWGN channel with $E_b/N_0 = 5.2$dB.]{\includegraphics[width=0.4\textwidth]{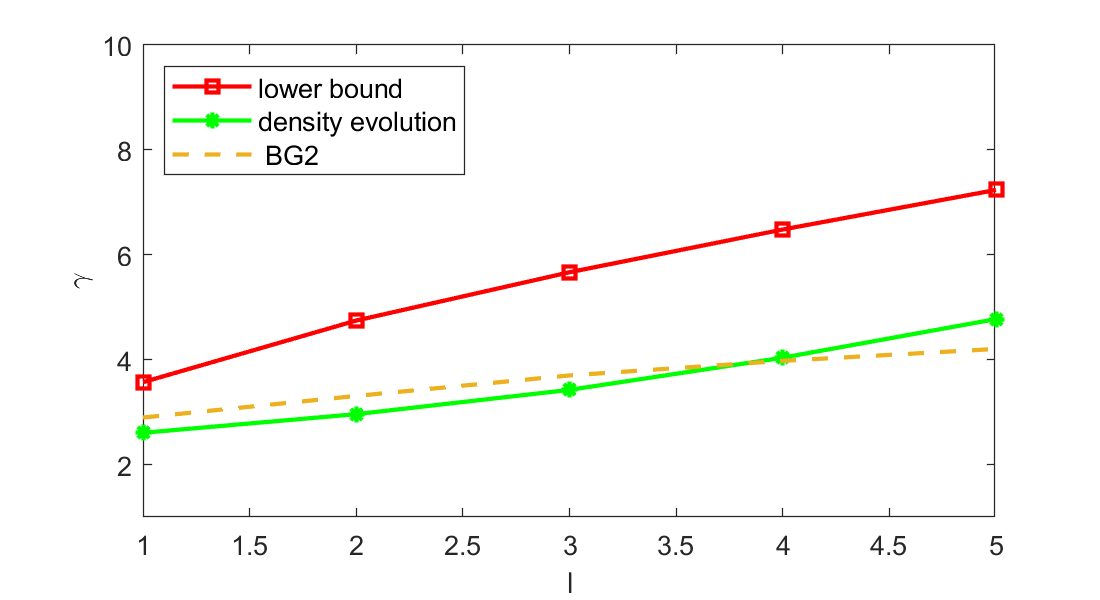}
\label{fig4_4}}
\caption{$\gamma$ vs. $J$ for the bound of LDPC code ensembles, illustrating the tightness of the proposed lower bound. }
\label{fig_4}
\end{figure}

\section{conclusion}
\label{sec4}

In this study, we establish a closed-form lower bound on the BER of LDPC codes as a function of the number of iterations $l$ under BP decoding. Specifically, for the regular LDPC code ensemble LDPC$(N,x^J,x^K)$, we show that the average BER admits a lower bound of the form 
\begin{equation}
    h2^{-a2^{bl}}
\end{equation}
for the BSC, the BEC, and the BI-AWGN channel, where $b=\log_2(J-1)$ is the main term determining the order of performance when $l \le \theta_1\log_{(J-1)^5(K-1)^3}N^2 $ with $\theta_1<1$. By comparing with the conclusions in \cite{lentmaier2005analysis}, we demonstrate that the lower bound we provided is tight in terms of order. Furthermore, for irregular LDPC code ensembles, we also propose a method, distinct from the regular case, to iteratively give the lower bound of the average BER.

These results offer insights into the fundamental decoding limits of LDPC codes and provide theoretical guidance for the design and optimization of practical coding systems. In particular, under a fixed decoding complexity, decoding strategies that enable each variable node to receive a greater number of channel messages are generally more effective. Such structures should therefore be prioritized in the design of encoding and decoding algorithms. For instance, compared to flooding BP, layered BP with various scheduling strategies enables faster propagation of the messages and has been widely adopted in practical implementations \cite{jang2022design,chang2023optimization}.

Regarding future work, we believe that the lower bounds on the BER of additional code ensembles, such as the ensembles of QC-LDPC and spatially coupled LDPC codes, can be investigated based on the number of distinct variable nodes in the computation graph. This will contribute to the development of more efficient encoding and decoding schemes under high-throughput conditions, thereby providing valuable guidance for practical applications.

\begin{appendices}

\section{PROOF OF LEMMA \ref{lemma_regular_w}}
\label{prf_BER}
Consider a Tanner graph \(G\) drawn uniformly at random from the LDPC ensemble \(\text{LDPC}(N, x^J, x^K)\), and a variable node \(v\) selected uniformly at random from \(G\). The following lemma shows that $w_{2l}(v, G)$ is, with high probability, less than or equal to $\frac{J(J-1)^l - 2}{J - 2}$, which in turn provides an upper bound on the expected value of $w_{2l}(v, G)$.

\begin{lemma}
\label{appendix_lemma}
When $l\le \log_{(J-1)(K-1)}\left(\left(1-\frac{K}{J(K-1)}\right)N\right)$, the probability $\mathring{P}$ that $w_{2l}(v, G) \le \frac{J(J-1)^l - 2}{J - 2}$ satisfies
\begin{equation}
\mathring{P} \ge \left[1-\frac{m_{2l}^*+C_{2l+2}}{M}
  \right]^{C_{2l+2}}\left[1-\frac{n_{2l}^*+V_{2l+2}}{N}
  \right]^{V_{2l+2}}.
\end{equation}
where 
\begin{equation}
  V_{2l+2}=\frac{2-J(J-1)^{l+1}}{2-J},  
\end{equation}
\begin{equation}
 C_{2l+2} = \frac{J-J(J-1)^{l+1}}{2-J}, 
\end{equation}
\begin{equation}
 n_{2l}^*= \frac{J(K-1)^{l+1}(J-1)^{l}-K}{(J-1)(K-1)-1},   
\end{equation}
\begin{equation}
 m_{2l}^*= \frac{J(K-1)^{l}(J-1)^{l}-J}{(J-1)(K-1)-1},   
\end{equation}
and $M$ and $N$ denote the number of check nodes and variable nodes in $G$, respectively.
\end{lemma}

For the purpose of proving Lemma \ref{appendix_lemma}, we first introduce some key definitions. We adopt a more suitable representation of the computation graph in which each distinct node appears only once \cite[Chapter 3.7.1]{richardson2008modern}. Specifically, with a slight abuse of notation, a computation graph $T_{k}(v,G)$ of height $k \ge 0$, rooted at a variable node $v$, is defined as the subgraph of the Tanner graph $G$ induced by the nodes at a distance less than $k$ from $v$ together with all edges among these nodes. A node is said to be at level $k^
\prime$ in the computation graph $T_{k}(v,G)$ if its distance from the root node $v$ is exactly $k^
\prime$. As illustrated in Fig.~\ref{fig_bec_gap}, the computation graph may either be a tree or contain cycles.

\begin{figure}[!t]
\centering
\subfigure[The computation graph is tree-like.]{
\includegraphics[width=3in]{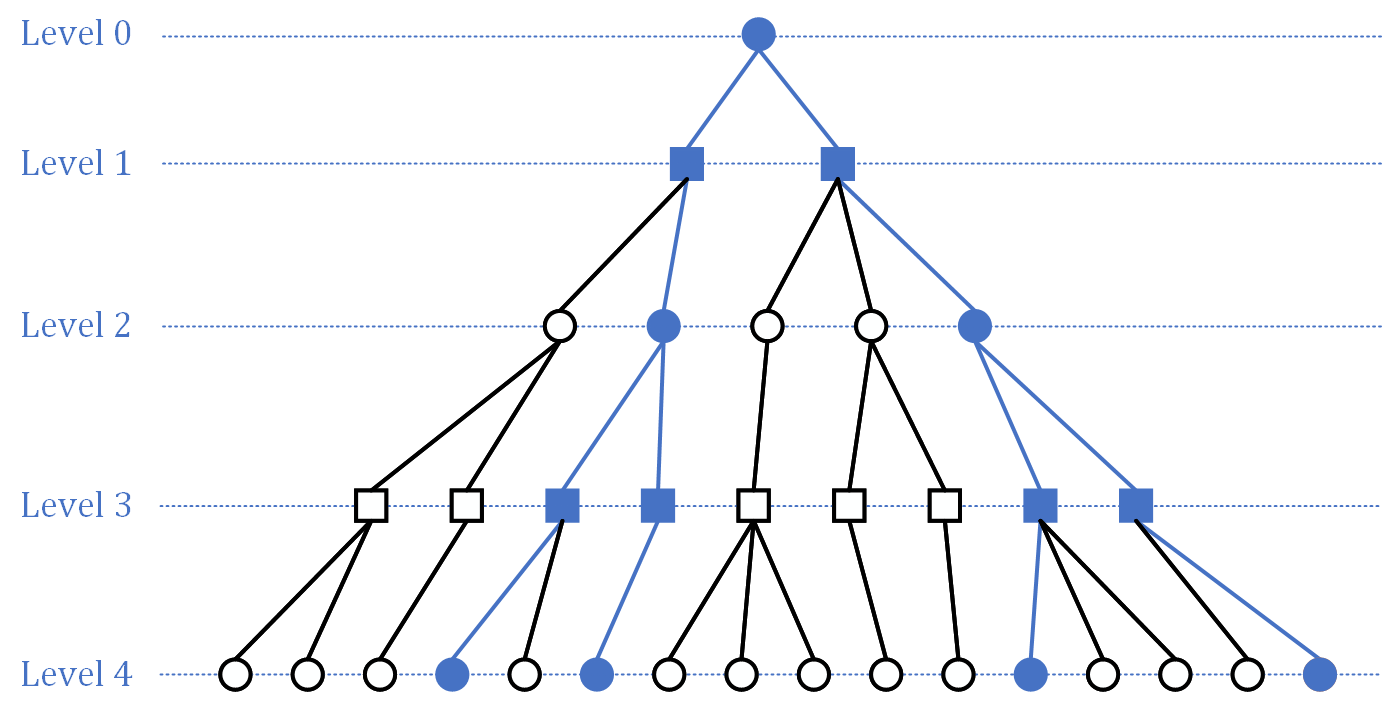} 
\label{min1}
}
\subfigure[The computation graph contains cycles. ]{
\includegraphics[width=3in]{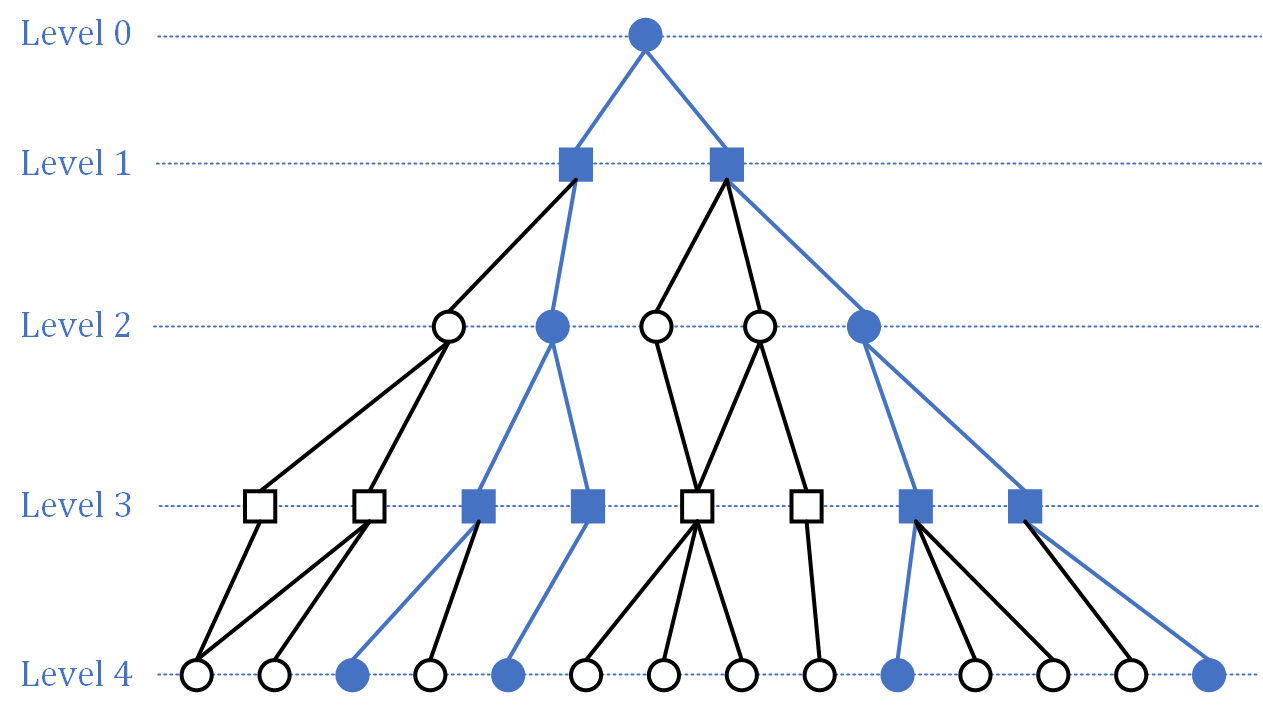} 
\label{min2}
}
\DeclareGraphicsExtensions.
\caption{The computation graphs and the valid trees (in solid blue) contained within them. Here, the computation graph differs from the one in the body of the paper, as repeated nodes appear only once.}
\label{fig_bec_gap}
\end{figure}


\begin{definition}
Given a computation graph $T_{k}(v,G)$ of height $k\ge 0$. A valid tree of $T_{k}(v,G)$, denoted by $VT_k(v,G)$, is a connected tree-like subgraph of height $k$ that satisfies the following conditions:
\begin{itemize}
    \item $v$ is contained in \(VT_k(v, G)\).
    \item For any variable node \(v'\) located at level $k_{v'}<k$ in \(T_k(v, G)\), if $v'$ is contained in \(VT_k(v, G)\), then all of its neighboring nodes are also contained in \(VT_k(v, G)\) and $v'$ has at most one neighbor situated at level $k_{v'} -1$ in $T_{k}(v,G)$.
    \item For any check node \( c \) located at level \( k_c<k \) in \( T_k(v, G) \), if \( c \) is contained in \( VT_k(v, G) \), then it has exactly two neighbors in \( VT_k(v, G) \), one situated at level \( k_c - 1 \) and the other at level \( k_c + 1 \) of \( T_k(v, G) \).
    \item For any node \( u \) located at level \( k \) in \( T_k(v, G) \), if \( u \) is contained in \( VT_k(v, G) \), then it has exactly one neighbor in \( VT_k(v, G) \) located at level \( k - 1 \).
\end{itemize}
\end{definition}

Fig.~\ref{fig_bec_gap} illustrates a valid tree, highlighted in solid blue. By definition, a computation graph may contain multiple valid trees. However, we are only concerned with whether the computation graph contains at least one valid tree and with the number of nodes within the valid tree, which remains invariant. Specifically, the valid tree $VT_k(v,G)$ is a tree with height $k$ in which every check node at levels less than $k$ has degree 2, and every variable node at levels less than $k$ has degree $J$. It follows that $VT_k(v,G)$, contains a total of 
\begin{equation}
  V_{k}=\frac{2-J(J-1)^{\lfloor k/2\rfloor}}{2-J}  
\end{equation}
variable nodes and 
\begin{equation}
C_{k} = \frac{J-J(J-1)^{\lfloor (k+1)/2\rfloor}}{2-J}   
\end{equation}
check nodes. 

Let $n_{k}$ and $m_{k}$ denote the total number of distinct variable nodes and check nodes, respectively, in the computation graph $T_{k}(v,G)$. Then, we have
\begin{equation}
 n_{k} \le n_{k}^*\triangleq \frac{J(K-1)^{{\lfloor k/2\rfloor}+1}(J-1)^{\lfloor k/2\rfloor}-K}{(J-1)(K-1)-1},   
\end{equation}
\begin{equation}
  m_{k} \le m_{k}^*\triangleq \frac{J(K-1)^{\lfloor (k+1)/2\rfloor}(J-1)^{\lfloor (k+1)/2\rfloor}-J}{(J-1)(K-1)-1}.  
\end{equation}
We provide a brief proof of the above equations. Since the degree of $v$ is $J$, the number of nodes at the first level of the computation graph is $J$. Similarly, for each check node at the first level, there are at most $K-1$ distinct neighboring variable nodes at level two. Therefore, the total number of distinct variable nodes at level two is bounded above by $J(K-1)$. Following this pattern, at the $k^\prime$-th level, there are at most $J(J-1)^{{{\lfloor k^\prime/2\rfloor}}-1}(K-1)^{{\lfloor k^\prime/2\rfloor}}$ distinct variable nodes. This implies that $n_{k} \le n_{k}^*$ and $m_{k} \le m_{k}^*$.

We now return to the proof of Lemma \ref{appendix_lemma}.
\begin{proof}[Proof of Lemma \ref{appendix_lemma}]
Assume that the computation graph $T_{2t}(v,G)$ contains a valid tree $VT_{2t}(v,G)$, and denote this event by $\mathcal{A}_{2t}$, where $t\ge 0$ is an integer constant. We now estimate the conditional probability $\mathbb{P}(\mathcal{A}_{2t+1} \mid \mathcal{A}_{2t})$, that is, the probability that the computation graph $T_{2t+1}(v,G)$ contains a valid tree $VT_{2t+1}(v,G)$, given that $\mathcal{A}_{2t}$ has occurred. This conditional probability $\mathbb{P}(\mathcal{A}_{2t+1} \mid \mathcal{A}_{2t})$ is greater than the probability that the following two conditions are simultaneously satisfied: (i) For any variable node at level $2t$ in $T_{2t}(v,G)$, if it is contained in $VT_{2t}(v,G)$, then its degree in $T_{2t}(v,G)$ is one, and (ii) the subgraph of $T_{2t+1}(v, G)$ induced by all nodes in $VT_{2t}(v,G)$ together with the check node at level $2t+1$ remains cycle-free. To lower-bound the probability $\mathbb{P}(\mathcal{A}_{2t+1} \mid \mathcal{A}_{2t})$, we assume that the tree $VT_{2t}(v,G)$ and the computation graph $T_{2t-1}(v,G)$ are already fixed. We then sequentially reveal the remaining $J-1$ edges of each variable node at level $2t$ in $VT_{2t}(v,G)$. Suppose that $E_c < C_{2t+2} - C_{2t}$ such edges have already been revealed, with each edge connected to a check node that is not yet present in $T_{2t-1}(v,G)$ and does not form a cycle. Then, the probability that the $(E_c+1)$-th edge also connects to a new check node and does not create a cycle is lower bounded by
\begin{equation}
\begin{split}
 &\frac{MK-m_{2t}K-E_cK}{MK}=1-\frac{E_c+m_{2t}}{M}.
\end{split}    
\end{equation}
As a result, the probability $\mathbb{P}(\mathcal{A}_{2t+1} \mid \mathcal{A}_{2t})$ is lower bounded by
\begin{equation}
 \begin{split}
 &\prod_{E_c=0}^{C_{2t+2}-C_{2t}-1}1-\frac{E_c+m_{2t}}{M} \\
  &\ge \left[1-\frac{C_{2t+2}-C_{2t}-1+m_{2t}}{M}
  \right]^{C_{2t+2}-C_{2t}} \\
  &\ge \left[1-\frac{C_{2t+2}+m^*_{2t}}{M}
  \right]^{C_{2t+2}-C_{2t}}.
\end{split}   
\end{equation}

Second, We next estimate the probability $\mathbb{P}(\mathcal{A}_{2t+2} \mid \mathcal{A}_{2t+1})$. This conditional probability equals the probability that the following two conditions are simultaneously satisfied: (i) for any check node at level $2t+1$ in $T_{2t+1}(v,G)$, if it is contained in $VT_{2t+1}(v,G)$, then its degree at most $K-1$; and (ii) the subgraph of $T_{2t+2}(v,G)$ induced by all nodes in $VT_{2t+1}(v,G)$ together with the variable node at level $2t+2$ remains cycle-free. Similar to the previous case, we assume that the tree $VT_{2t+1}(v,G)$ and the computation graph $T_{2t}(v,G)$ are already fixed. We then sequentially reveal one edge for each check node at level $2t+1$ in $VT_{2t+1}(v,G)$, and bound the probability that each newly revealed edge connects to a variable node not yet present in $T_{2t}(v,G)$ and does not create a cycle. As a result, the probability $\mathbb{P}(\mathcal{A}_{2t+2} \mid \mathcal{A}_{2t+1})$ can be lower bounded by
\begin{equation}
     \begin{split}
  &\prod_{E_v=0}^{V_{2t+2}-V_{2t}-1}\frac{NJ-n_{2t}J-E_vJ}{NJ}\\
 &\ge \prod_{E_v=0}^{V_{2t+2}-V_{2t}-1} 1-\frac{n_{2t}+E_v}{N} \\
  &\ge \left[1-\frac{n_{2t}^*+V_{2t+2}}{N}
  \right]^{V_{2t+2}-V_{2t}},
\end{split}
\end{equation}

In summary, the probability that the computation graph $T_{2l+1}(v,G)$ contains a valid tree rooted at $v$ of height $2l+1$ is lower bounded by
\begin{equation}
    \begin{split}
&\prod_{t=0}^{l}\mathbb{P}(\mathcal{A}_{2t+1}\big|\mathcal{A}_{2t})\mathbb{P}(\mathcal{A}_{2t+2}\big|\mathcal{A}_{2t+1}) \\
    &=\prod_{t=0}^{l} \left[1-\frac{m_{2l}^*+C_{2t+2}}{M}
  \right]^{C_{2t+2}-C_{2t}} \left[1-\frac{n_{2t}^*+V_{2t+2}}{N}
  \right]^{V_{2t+2}-V_{2t}}\\
  &\ge\left[1-\frac{m_{2l}^*+C_{2l+2}}{M}
  \right]^{C_{2l+2}}\left[1-\frac{n_{2l}^*+V_{2l+2}}{N}
  \right]^{V_{2l+2}}.\\
\end{split}
\end{equation}

When the computation graph $T_{2l+1}(v,G)$ contains a valid tree $VT_{2l+1}(v,G)$, a codeword in $\mathcal{C}_{2l}^{1}(v,G)$ can be obtained by assigning the value 1 to all variable nodes in the valid tree and 0 to all remaining variable nodes. The resulting codewords has Hamming weight of exactly \(\frac{J(J-1)^l - 2}{J - 2}.\) Consequently, an upper bound on $w_{2l}^1(v,G)$ is given by $w_{2l}^1(v,G) \le \frac{J(J-1)^l - 2}{J - 2}$. 
\end{proof}

We now return to the proof of Lemma \ref{lemma_regular_w}.

\begin{proof}[Proof of Lemma \ref{lemma_regular_w}]

When the computation graph $T_{2l+1}(v,G)$ does not contain a valid tree, the minimum Hamming weight $w_{2l}^1(v,G)$ is upper bounded by the number of variable nodes in $T_{2l}(v,G)$. That is,
\begin{equation}
 w_{2l}^1(v,G)\le n_{2l} \le n_{2l}^*. 
\end{equation}

According to Lemma \ref{appendix_lemma}, when $l\le \log_{(J-1)(K-1)}\left(\left(1-\frac{K}{J(K-1)}\right)N\right)$, the expected minimum Hamming weight of $\mathcal{C}_{{2l}}^1(v,G)$ is less than 
\begin{equation}
 \mathring{P}\frac{J(J-1)^l - 2}{J - 2} + (1-\mathring{P})*n_{2l}^*.   
\end{equation}

This implies that when $l \le \theta_1 \log_{(J-1)^5(K-1)^3} N^2$ with $\theta_1 < 1$ and $N$ is sufficiently large, we have the expected value of $w_{2l}^1(v,G)$ approaches \ $\frac{J(J - 1)^l-2}{J - 2}$.
\end{proof}

\section{PROOF OF LEMMA \ref{lemma_irregular_w} }
\label{appendix_irregular_w}
Consider a Tanner graph \(G\) drawn uniformly at random from the LDPC ensemble \(\text{LDPC}(N, L, x^2)\), and let variable nodes \(v_i\) and $v_j$ be chosen independently and uniformly at random from \(G\). Let \(n_{2l}(v_i,G)\) denote the number of distinct variable nodes in \(T_{2l}(v_i,G)\), and let \(\overline{n}_{2l}(N,L,x^2)\) denote the ensemble average of \(n_{2l}(v_i,G)\).

According to Lemma~\ref{lemma_regular_w}, when \(
l \le \theta_1\log_{(J_{\max}-1)^5(K_{\max}-1)^3} N^2,\theta_1<1,\) the expected minimum Hamming weight $\overline{w}_{2l}^1(N,L,R)$ is upper-bounded by $\overline{n}_{2l}(N,L,x^2)$. It suffices to provide an upper bound on $\overline{n}_{2l}(N,L,x^2)$. Recall that $dist(v_i,v_j)$ denotes the graph distance between $v_i$ and $v_j$. Then $n_{2l}(v_i,G)$ equals the number of variable nodes at distance less than $2l$ from $v_i$, and thus

 \begin{equation}
 \begin{split}
 \label{e1}
 &\overline{n}_{2l}(N,L,x^2) = N\Big(1 - \mathbb{P}(dist(v_i,v_j) > 2l)\Big) \\
 & = N\left(1 - \mathbb{P}\left(dist(v_i,v_j)>0\right)\prod_{t^\prime=1}^{2l} \mathbb{P}\left(dist(v_i,v_j) >t^\prime|dist(v_i,v_j)>t^\prime-1 \right)\right)\\
 & \overset{(e5)}{=}  N\Bigg(1 -\mathbb{P}\left(dist(v_i,v_j)>0\right)\prod_{t=1}^{l} \mathbb{P}\left(dist(v_i,v_j) >2t|dist(v_i,v_j)>2t-1 \right)\Bigg)  .
 \end{split}
 \end{equation}
 where $\mathbb{P}(dist(v_i,v_j) > 2t|dist(v_i,v_j) > 2t-1)$ denotes the conditional probability that the distance $dist(v_i,v_j)$ is greater than $2t$ given that it is greater than $2t-1$, and is hereafter abbreviated as $P_{2t}$. The validity of $(e5)$ follows from the fact that, in a bipartite graph, the path length between any two variable nodes is always even. Therefore, $\mathbb{P}(dist(v_i,v_j) > 2t + 1 \mid dist(v_i,v_j) > 2t) = 1$ for $t = 0, \ldots, l-1$.
 
 In the following, we adopt a recursive approach to compute $P_{2t}$, inspired by the method introduced in \cite{nitzan2016distance}, which analyzes the distribution of shortest path lengths in random graphs. 
\subsection{some recursive expressions in a Tanner graph}
\label{sub1}

In this Subsection, we introduce the definition of the mean conditional indicator function and demonstrate that it can be computed via a sequence of recursive expressions. These recursive properties play a central role in the computation of $P_{2t}$.

\begin{definition}
The definitions of the indicator functions are as follows:
\begin{equation}
    \mathcal{X}(dist(v_i,v_j) > 2t) = \left\{ 
\begin{array}{lr}
1 & dist(v_i,v_j) >2t, \\
0 & dist(v_i,v_j) \le 2t.
\end{array}
\right.
\end{equation}

The definitions of the conditional indicator functions are as follows:
\begin{equation}
\begin{split}
     &\mathcal{X}(dist(v_i,v_j) > 2t|dist(v_i,v_j) > 2t - 1)= \frac{\mathcal{X}(dist(v_i,v_j) > 2t) \cap (dist(v_i,v_j)>2t-1)}{\mathcal{X}(dist(v_i,v_j) > 2t-1)}.
\end{split}
\end{equation}
\end{definition}

Note that, $\mathcal{X}(dist(v_i,v_j) > 2t|dist(v_i,v_j) > 2t - 1)$  indicates whether $dist(v_i,v_j)$ is greater than $2t$, given that $dist(v_i,v_j)$ is greater than $2t-1$. If this is true, $\mathcal{X}(dist(v_i,v_j) > 2t|dist(v_i,v_j) > 2t - 1)$ takes a value of $1$; otherwise, it assumes a value of $0$. In case where the condition $dist(v_i,v_j) > 2t-1$ is not satisfied, the value of the conditional indicator function is undetermined. Though our primary focus is on the conditional indicator functions between pairs of variable nodes, the same definitions apply equally to pairs of check nodes as well as to variable node and check node pairs.

\begin{definition}
If we take the average of the conditional indicator function $\mathcal{X}(dist(v_i,v_j) > 2t|dist(v_i,v_j) > 2t - 1)$ over all suitable choices of $v_j$, we obtain the mean conditional indicator function $I_{v_i}(2t)$:
\begin{equation}
    I_{v_i}(2t) = \mathbb{E}_{v_j}\left[\mathcal{X}(dist(v_i,v_j) > 2t|dist(v_i,v_j) > 2t - 1)\right].
\end{equation}
The averaging is done only over nodes $v_j$ for which the condition $dist(v_i,v_j) > 2t-1$ is satisfied.
\end{definition}


To establish the recursive expression for the mean conditional indicator function, we first examine the structural properties of paths in the Tanner graph. As illustrated in Fig. \ref{fig5}, a path of length $2t$ from a node $v_i$ to a node $v_j$ can be decomposed into an edge from $v_i$ to $c_r \in \mathcal{N}(v_i)$, and a path of length $2t-1$ from $c_r$ to $v_j$. In other words, if there is no path of length $2t$ between node $v_i$ and node $v_j$, it implies that any neighbor $c_r$ of node $v_i$ does not have a path of length $2t-1$ to reach node $v_j$. The graph obtained by removing vertex $v_i$ from graph $G$ is referred to as the cavity graph of $G$, denoted as $\widetilde{G}$. Note that those paths from $c_r$ to $v_j$ should be embedded in $\widetilde{G}$, meaning that these paths should not pass through $v_i$ again. We refer to $\mathcal{X}^{(v_i)}(dist(c_r,v_j)>2t-1|dist(c_r,v_j)>2t-2)$ as the cavity indicator function, which indicates whether $dist(c_r,v_j)$ is greater than $2t-1$, given that $dist(c_r,v_j) > 2t-2$. The superscript $(v_i)$ stands for the fact that the node $c_r$ is reached by a link from node $v_i$. 

Drawing upon the recursive property of paths, it can be deduced that the conditional indicator function $\mathcal{X}(dist(v_i,v_j) > 2t|dist(v_i,v_j) > 2t - 1)$ can be denoted as the product of cavity indicator functions between nodes $c_r \in \mathcal{N}(v_i)$ and $v_j$:
\begin{equation}
\label{eq6_1}
\begin{split}
    &\mathcal{X}(dist(v_i,v_j) > 2t|dist(v_i,v_j) > 2t - 1)= \prod_{c_r \in \mathcal{N}(v_i)} \mathcal{X}^{(v_i)}(dist(c_r,v_j)>2t-1|dist(c_r,v_j)>2t-2).
\end{split}
\end{equation}
Based on the definition of the mean conditional indicator function $I_{v_i}(2t)$ and (\ref{eq6_1}), we can derive the recursive expression for the mean conditional indicator function. That is,
\begin{equation}
\label{eq8}
\begin{split}
&I_{v_i}(2t) \\
&= \mathbb{E}_{v_j}\left[\prod_{c_r \in \mathcal{N}(v_i)} \mathcal{X}^{(v_i)}(dist(c_r,v_j)>2t-1\mid dist(c_r,v_j)>2t-2)\right]\\
&\overset{(e6)}{=}\prod_{c_r \in \mathcal{N}(v_i)} \mathbb{E}_{v_j}\left[ \mathcal{X}^{(v_i)}(dist(c_r,v_j)>2t-1\mid dist(c_r,v_j)>2t-2)\right]\\
&=\prod_{c_r \in \mathcal{N}(v_i)} I_{c_r}^{(v_i)}(2t-1)
\end{split}
\end{equation}
where the cavity mean indicator function $I_{c_r}^{(v_i)}(2t-1)$ denotes the average of the cavity indicator functions $\mathcal{X}^{(v_i)}(dist(c_r,v_j)>2t-1|dist(c_r,v_j)>2t-2)$ over all suitable choices of $v_j$. The validity of equation $(e6)$ relies on the assumption that the local structure of $G$ is tree-like and $N$ is sufficiently large.

Moreover, $I_{c_r}^{(v_i)}(2t-1)$ follows a similar recursive structure, allowing the computation to proceed recursively.
\begin{equation}
\label{eq9}
    I_{c_r}^{(v_i)}(2t-1) = \prod_{v_{i^\prime} \in \mathcal{N}(c_r) \backslash \{v_i\} }I_{v_{i^\prime}}^{(c_r)}(2t-2).
\end{equation}

\begin{figure*}[!t]
\centering
\includegraphics[width=6in]{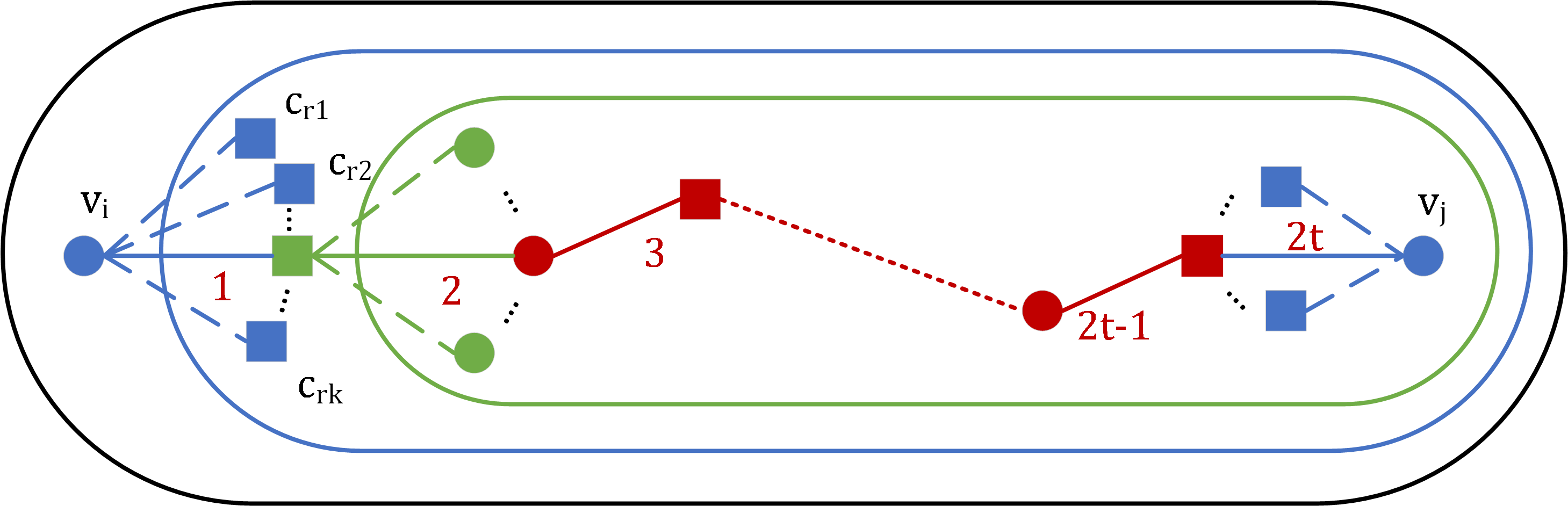}
\caption{Illustration of the possible paths of length $2t$ between two random variable nodes, $v_i$ and $v_j$, in a Tanner graph $G$. The first edge of such a path connects node $v_i$ to some other node, $c_r$, which may be any one of the $d$ neighbors of node $v_i$. The rest of the path, from node $c_r$ to node $v_j$ is of length $2t-1$ and it resides on the cavity graph of graph $G$.}
\label{fig5}
\end{figure*}

\subsection{recursive expression for $P_{2t}$}
\label{sub2}

With a probability of $L_d$, node $v_i$ has a degree of $d$. If $v_i$ is connected to a check node $c_r$ by an edge, and $c_r$ is directly connected to another variable node $v_{i^\prime}$ (with $v_{i^\prime} \neq v_i$), then the probability that $v_{i^\prime}$ has degree $d$ is given by $\lambda_d = \frac{dL_d}{\sum_{d=1}^\infty dL_d}$.

Let $\mathbb{P}\big(I_{v_i}(2t) = I\big)$ denote the probability that the mean indicator function associated with variable node $v_i$ takes on the value $I$, $\widetilde{\mathbb{P}}\big(I_{c_r}^{(v_i)}(2t-1)=\widetilde{I}\big)$ denote the probability that the cavity mean indicator function associated with check node $c_r$ takes on the value $\widetilde{I}$. Based on (\ref{eq8}), we can derive:
\begin{equation}
\label{eq7}
\begin{split}
    &\mathbb{P}\Big(I_{v_i}(2t) = I\Big)= \sum_{d=1}^{\infty}L_d \int_0^1 \int_0^1 \ldots \int_0^1 \prod_{r=1}^d \widetilde{\mathbb{P}}\Big(I_{c_r}^{(v_i)}(2t-1) = I_r\Big) \delta\Big(I-\prod_{r=1}^d I_r\Big)\mathrm{d}I_1 \mathrm{d}I_2 \ldots \mathrm{d}I_d,
\end{split}
\end{equation}
where
\begin{equation}
    \delta\Big(I-\prod_{r=1}^d I_r\Big) = \left\{
\begin{array}{lr}
1 & \text{if }I = \prod_{r=1}^{d}I_r, \\
0 & \text{if } I \neq \prod_{r=1}^{d}I_r.
\end{array}
\right.
\end{equation}

When the degree of node $v_i$ is $d$, $I_{v_i}(2t)$ is equal to $I$ if and only if the product of cavity mean indicator functions corresponding to the $d$ neighbors of node $v_i$ equals $I$. In other words, $\mathbb{P}\left(I_{v_i}(2t) = I\right)=\widetilde{\mathbb{P}}(\prod_{r=1}^d I_{c_r}^{(v_i)}(2t-1) = I) $. According to the local tree-like structure of $G$, $\widetilde{\mathbb{P}}(\prod_{r=1}^d I_{c_r}^{(v_i)}(2t-1) = I)$ can be expressed by the integral term in (\ref{eq7}). And the $\delta$ function constrains the integration domain to the region where $I=\prod_{r=1}^d I_{c_r}^{(v_i)}(2t-1)$. The validity of (\ref{eq7}) relies on the probability of a node $v_i$ having a degree $d$ being $L_d$.

Similarly, based on (\ref{eq9}), the following expression holds:
\begin{equation}
\label{eq10}
\begin{split}
    &\widetilde{\mathbb{P}}\Big(I_{c_r}^{(v_i)}(2t-1) = I\Big)=\widetilde{\mathbb{P}}\Big(I_{v_{i^\prime}}^{(c_r)}(2t-2) = I\Big),
\end{split}
\end{equation}
and

\begin{equation}
\label{eq11}
\begin{split}
    &\widetilde{\mathbb{P}}\Big(I_{v_{i^\prime}}^{(c_r)}(2t-2) = I\Big)= \sum_{d=1}^{\infty}  \lambda_d \int_0^1 \ldots \int_0^1 \prod_{r^\prime=1}^{d-1} \widetilde{\mathbb{P}}\Big(I_{c_{r^\prime}}^{(v_{i^\prime})}(2t-3) = I_{r^\prime}\Big)\delta\Big(I-\prod_{r^\prime=1}^{d-1} I_{r^\prime}\Big)\mathrm{d}I_1 \ldots \mathrm{d}I_{d-1}.
\end{split}
\end{equation}
Since the degree of the check node is 2, it follows that $I_{c_r}^{(v_i)}(2t-1) = I_{v_{i^\prime}}^{(c_r)}(2t-2)$, which leads to (\ref{eq10}). (\ref{eq11}) calculates the probability of the cavity mean indicator function $I_{v_{i^\prime}}^{(c_r)}(2t-2)$ taking the value $I$. For node $v_{i^\prime} \in \mathcal{N}(c_r)$, the probability of its degree being equal to $d$ is denoted as $\lambda_d$, and as an intermediate node, one of its edges is consumed by the incoming link, leaving only $d-1$ links for the outgoing paths. Therefore, this equation slightly differs from (\ref{eq7}).
\begin{figure}[!t]
\centering
\includegraphics[width=3.5in]{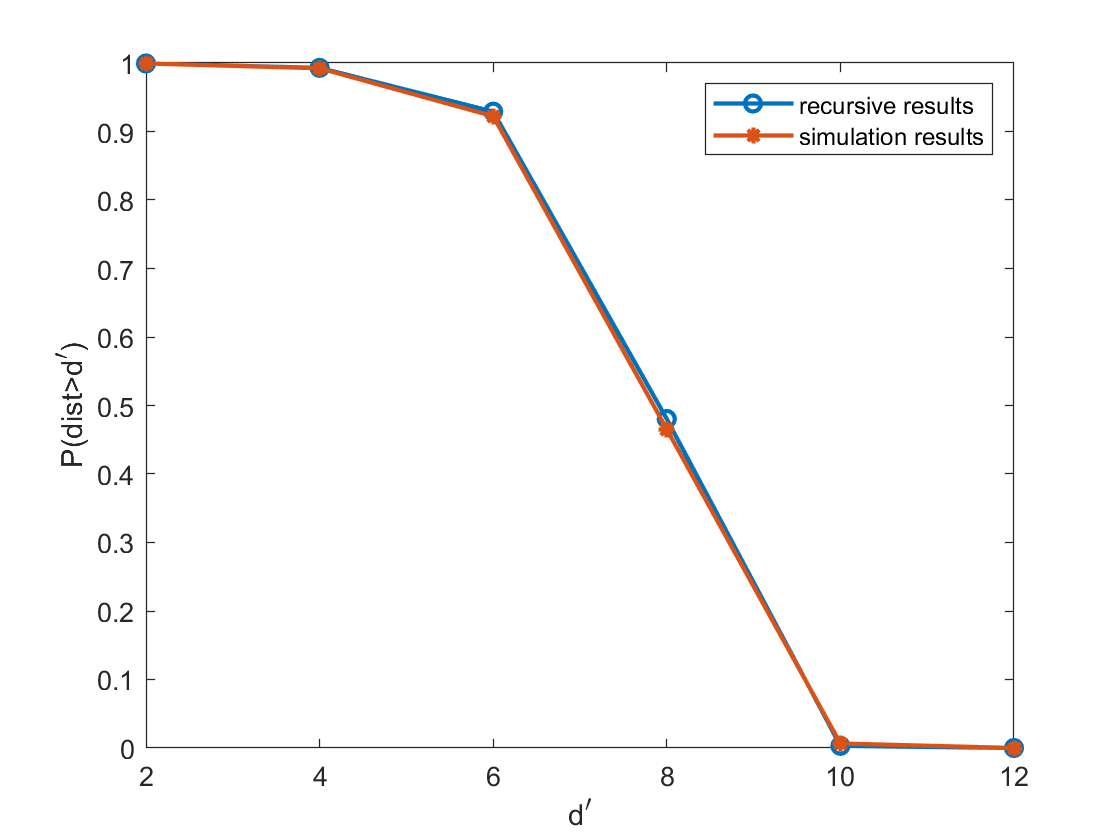}
\caption{The tail distribution $\mathbb{P}(dist(v_i,v_j) > d^\prime)$ were obtained from both recursive equations and numerical simulations. The calculations were performed on an irregular LDPC code ensemble with a block length of $N = 20000$. The degree distributions for the irregular LDPC code ensemble are specified as $\lambda(x) = 0.38354x + 0.04237x^2 + 0.57409x^3$ and $\rho(x) = 0.24123x^4 + 0.75877x^5$. The numerical results were averaged over 50 graph instances.}
\label{fig4}
\end{figure}

The expected values of $I_{v_i}(2t)$, $I_{c_r}^{(v_i)}(2t-1)$, and $I_{v_{i^\prime}}^{(c_r)}(2t-2)$ are denoted by $P_{2t}$, $\widetilde{P}_{2t-1}$, and $\widetilde{P}_{2t-2}$, respectively. That is,

\begin{equation}
\label{eq12}
\begin{split}
 &P_{2t} = \mathbb{P}(dist(v_i,v_j)>2t|dist(v_i,v_j)>2t-1)= \int_0^1 I\mathbb{P}\Big(I_{v_i}(2t) = I\Big) \mathrm{d}I,   
\end{split}
\end{equation}

\begin{equation}
\label{eq13}
\begin{split}
 &\widetilde{P}_{2t-1}= \int_0^1 I \widetilde{\mathbb{P}}\Big(I_{c_r}^{(v_i)}(2t-1)=I\Big)\mathrm{d}I   
\end{split}
\end{equation}
and
\begin{equation}
\label{eq14}
\begin{split}
&\widetilde{P}_{2t-2}=\int_0^1 I \widetilde{\mathbb{P}}\Big(I_{v_{i^\prime}}^{(c_r)}(2t-2)=I\Big)\mathrm{d}I.    
\end{split}
\end{equation}

Plugging Eqs. (\ref{eq7}), (\ref{eq10}) and (\ref{eq11}) into Eqs. (\ref{eq12}), (\ref{eq13}) and (\ref{eq14}), respectively, we obtain the recursion
equations
\begin{equation}
\label{irr_eqp2t}
P_{2t}=\sum_{d=1}^\infty L_d \bigg(\widetilde{P}_{2t-1}\bigg)^d, 
\end{equation}

\begin{equation}
\widetilde{P}_{2t-1} = \widetilde{P}_{2t-2}, 
\end{equation}
and 
\begin{equation}
\widetilde{P}_{2t-2} = \sum_{d=1}^\infty \lambda_d \bigg(\widetilde{P}_{2t-3}\bigg)^{d-1}.
\end{equation}
When the number of nodes is sufficiently large, we can approximately obtain
\begin{equation}
\label{irr_1}
\widetilde{P}_1 = 1-\frac{1}{N}.
\end{equation} This completes the proof of Lemma \ref{lemma_irregular_w} for the case $l \le \theta_1 \log_{(J_{\max}-1)^5(K_{\max}-1)^3} N^2$ with $\theta_1 <1$.

When $l > \theta_1\log_{(J_{\max}-1)^5(K_{\max}-1)^3} N^2$, one can obtain an upper bound on $\overline{w}_{2l}^1(N,L,R)$ by using the expected number of distinct variable nodes in the computation graph generated from the ensemble $\mathrm{LDPC}(N, L, R)$. The expected number of distinct variable nodes can be computed analogously and is given by
\begin{equation}
    N\left(1 - \prod_{t=1}^{l} P_{2t}\right),
\end{equation}
where, 
\begin{equation}
\label{irrl_eqp2t}
P_{2t}=\sum_{d=1}^\infty L_d \bigg(\widetilde{P}_{2t-1}\bigg)^d, 
\end{equation}
\begin{equation}
\widetilde{P}_{2t-1} = \sum_{d=1}^\infty \rho_d \left(\widetilde{P}_{2t-2}\right)^{d-1}, 
\end{equation}
\begin{equation}
\widetilde{P}_{2t-2} = \sum_{d=1}^\infty \lambda_d \bigg(\widetilde{P}_{2t-3}\bigg)^{d-1}.
\end{equation}
and
\begin{equation}
\label{irrl_1}
\widetilde{P}_1 = \sum_{d=1}^\infty \rho_d \left(1-\frac{1}{N}\right)^{d-1}.
\end{equation}

In Fig. \ref{fig4}, we depict the tail distribution $\mathbb{P} (dist(v_i,v_j) > d^\prime)=\prod_{t=1}^{d^\prime} \mathbb{P}(dist(v_i,v_j)>t|dist(v_i,v_j)>t-1)$ for an ensemble of irregular LDPC codes with block length $N=20000$. This result is obtained using (\ref{irrl_eqp2t})$\sim$(\ref{irrl_1}). The degree distributions for the irregular LDPC code ensemble are specified as $\lambda(x) = 0.38354x + 0.04237x^2 + 0.57409x^3$ and $\rho(x) = 0.24123x^4 + 0.75877x^5$. The results are compared with computer simulations showing excellent agreement.
\end{appendices}

\bibliographystyle{ieeetr}
\bibliography{reference}

\end{document}